\documentclass[12pt]{article}

\usepackage{geometry}
\usepackage{graphicx}
\usepackage{subfigure}
\usepackage{float}

\usepackage{amsthm}
\usepackage{amsmath}
\usepackage{amssymb}
\usepackage{booktabs}

\usepackage[toc,page,header]{appendix}
\usepackage{minitoc}

\usepackage[natbib=true,style=apa,backend=biber]{biblatex}
\usepackage{bm}

\usepackage{algorithm}
\usepackage[noend]{algpseudocode}
\usepackage[table]{xcolor}

\RequirePackage[colorlinks,linkcolor=blue,citecolor=blue,urlcolor=blue]{hyperref}
\RequirePackage{cleveref}

\newcommand\balpha{{\bm \alpha}}
\newcommand\bn{{\bm n}}

\newcommand\bpi{{\bm \pi}}
\newcommand\mB{{\mathcal B}}
\newcommand\mN{{\mathcal N}}

\newcommand\R{{\mathbb R}}

\newcommand\Dir{{\rm{Dir}}}
\newcommand\PH{{\rm{PH}}}
\newcommand\Mu{{\rm{Mu}}}
\newcommand\Exp{{\rm{Exp}}}
\newcommand\Be{{\rm{Be}}}
\newcommand\Ga{{\rm{Ga}}}
\newcommand\Sg{{\rm{Sg}}}
\newcommand\DM{{\rm{DM}}}

\newcommand\NM{{\rm{NM}}}
\newcommand\E{{\mathbb{E}}}
\newcommand\PPH{{\rm{PPH}}}
\newcommand\NB{{\rm{NB}}}
\renewcommand\d{{\mathrm d}}

\newcommand{\abs}[1]{\left\vert#1\right\vert}
\newcommand{\norm}[1]{\|#1\|}

\newcommand\iid{\overset{\text{iid}}{\sim}}

\newtheorem{corollary}{Corollary}
\newtheorem{theorem}{Theorem}
\newtheorem{definition}{Definition}
\newtheorem{proposition}{Proposition}
\newtheorem{remark}{Remark}

\crefname{thm}{Theorem}{Theorems}
\crefname{prop}{Proposition}{Propositions}
\crefname{lem}{Lemma}{Lemmas}
\crefname{coro}{Corollary}{Corollaries}
\crefname{add}{Addendum}{Addendums}
\crefname{asm}{Assumption}{Assumptions}
\crefname{alg}{Algorithm}{Algorithms}
\crefname{proc}{Procedure}{Procedures}
\crefname{exe}{Exercise}{Exercises}
\crefname{exa}{Example}{Examples}
\crefname{prob}{Problem}{Problems}
\crefname{section}{Section}{Sections}
\crefname{subsection}{Section}{Sections}
\crefname{appendix}{Appendix}{Appendices}

\begin{document}

\title{Pochhammer Priors for  Sparse Count Models}

 \author{Yuexi Wang\footnote{Department of Statistics, University of Illinois Urbana-Champaign} \and  Nicholas G. Polson\footnote{Booth School of Business, University of Chicago}}

\date{1st Version: Febuary 13, 2024\\
This Version: \today}

\def\spacingset#1{\renewcommand{\baselinestretch}%
{#1}\small\normalsize} \spacingset{1.5}

\maketitle


\spacingset{1.1}

\doparttoc 
\faketableofcontents 


\begin{abstract}
\noindent   Bayesian hierarchical models are commonly employed for inference in count datasets, as they account for multiple levels of variation by incorporating prior distributions for parameters at different levels. Examples include Beta-Binomial, Negative-Binomial (NB), Dirichlet-Multinomial (DM) distributions.   In this paper, we address two crucial challenges that arise in various Bayesian count models:  inference for the concentration parameter in the ratio of Gamma functions and the inability of these models to effectively handle  excessive zeros and small nonzero counts. We propose a novel class of prior distributions that  facilitates conjugate updating of the concentration parameter in Gamma ratios, enabling full Bayesian inference for the aforementioned count distributions. We use DM models as our running examples. Our methodology leverages fast residue computation and admits closed-form posterior moments. 
 Additionally,  we recommend a default horseshoe type prior which has a heavy tail and substantial mass around zero. It admits continuous shrinkage, making the posterior highly adaptable  to sparsity or quasi-sparsity in the data.
 Furthermore, we offer insights and potential generalizations to other count models facing the two challenges. We demonstrate the usefulness of our approach on both simulated examples and on  real-world applications. Finally, we conclude with directions for future research. 
\end{abstract}

\vspace{0.15in}
    
\noindent%
\vfill

\spacingset{1.5} 

\begin{refsection}

\section{Introduction}
\vspace{-0.5cm}

 In this paper, we provide a conjugate family of prior distributions facilitating full Bayesian inference for various count models that possess a Gamma function ratio component.  We focus on two major computational challenges common to these count models. The first one concerns the inference of concentration parameter $\alpha$, which is notoriously hard due to the Gamma ratio structure. This structure leads to a non-concave log likelihood function and makes conjugate updates hard to construct. Thus, uncertainty quantification for $\alpha$ is inevitably complicated for both frequentist  and Bayesian approaches.  Second, these models cannot properly accommodate sparse or quasi-spase counts.   Zero-inflated count datasets are frequently encountered in real-world applications. For instance, in bag-of-words analyses, certain words exhibit frequent occurrences while others appear infrequently or not at all.  Similarly, in microbiome data \citep{zeng2022zero,koslovsky2023bayesian} and insurance claims, structural zeros are common, with many instances of no observed events. 
 
A related but less discussed phenomenon, quasi-sparsity, occurs when there are an overabundance of zero and small nonzero counts. This concept, to our knowledge, was first adopted in \citet{datta2016bayesian}, where they propose a continuous shrinkage prior for Poisson distribution, and has since been rarely studied for other count models. Existing methods specifically targeted at zero inflation are insufficient flexible for quasi-sparse count datasets. One example of quasi-sparsity is the  International Classification  of Diseases (ICD) code data from Electronic Health Records (EHRs), where topic models are commonly applied to such data for more efficient interpretation. It is common practice to remove low-frequency ICD codes or grouping infrequent ICD codes under broader diagnostic categories.  However, ad hoc preprocessing runs the risk of discarding meaningful signals that may be important for certain insights. 
%

While our proposed priors serve as an overarching solution for a range of count models, we focus on the DM model to provide a concrete example and then discuss how to generalize our results and intuitions in \Cref{sec:other_models}. We first outline our setup for the Dirichlet-Multinomial (DM)  distribution. Suppose we have $N$ items distributed across $K$ categories. The DM model assumes the count vector $\bn = (n_1, \dots, n_K)$ follows a multinomial distribution $\Mu(\bpi, N)$ with $N=\sum_{i=1}^K n_k$ and the corresponding vector of probabilities $\bpi = (\pi_1,  ..., \pi_K)$ follows a Dirichlet distribution $\Dir(\balpha)$ with hyperparameters $\balpha = (\alpha_1,  \ldots, \alpha_K)$.

One advantage of the Dirichlet prior is that the posterior of $\bpi$ also follows a Dirichlet distribution $\Dir(\bn+\balpha)$, with a convenient representation of the posterior mean as
\begin{equation}\label{eq:smoothing}
\E(\pi_k\mid \bn, \balpha)= \frac{ n_k + \alpha_k }{N + \sum_{k=1}^K \alpha_k},
\end{equation}
which corresponds to the probability estimator known as additive smoothing.   The conjugate nature of the Dirichlet distribution for the multinomial distribution and its ease of use have permitted good results in various modern applications, including  Latent Dirichlet Allocation (LDA)\citep{blei2003latent}, a popular topic modeling algorithm that is widely used in natural language processing.  

Despite the prevalence of DM models, it has been a long-standing challenge to select its concentration hyperparameters $\balpha$. The choice significantly impacts the posterior distribution of category probabilities and, consequently, the model's classification accuracy. Often the parameter is chosen in an ad-hoc manner, 
with common choices including the uniform prior $\alpha=1$, the Jeffreys prior $\alpha=1/2$, and  the ``Overall Objective Prior" $ \alpha =1/K$ recommend by \citet{berger2015overall}. Although additive smoothing is straightforward to implement, it can produce biased estimates and is not always effective in practice. Another approach to choosing $\balpha$ is via maximum marginal likelihood,  $\hat\balpha=\arg\max_{\balpha}p(\bn\mid \balpha)$, also known as the empirical Bayes choice. Unfortunately, $p( \bn\mid \balpha)$ is usually not analytically tractable. Due to the presence of complicated gamma ratios, the log likelihood function is likely to be non-concave, which makes optimization challenging.
Various approximation schemes have been  proposed to find the maximizer of the marginal distribution. The first type is to approximate the Dirichlet likelihood with a simpler function  by matching the first two derivatives \citep{minka2000estimating}. The same strategy was also used in inferring the Gamma shape parameter \citep{miller2019fast}. Another common strategy involves variants of the Expectation-Maximization  (EM) algorithm, including Monte Carlo EM  \citep{wallach2006topic} and variational EM \citep{blei2003latent}, and variants of Markov Chain Monte Carlo (MCMC) algorithms \citep{george2017principled,xia2020scalable}. 
 However, these methods only provide a point estimate for $\alpha$ and are unable to supply uncertainty quantification. 

A further issue arises from handling an abundance of zeros or small counts,  which leads to a discrepancy between the distribution suggested by the prescribed model and empirical distribution observed in the data. The additive smoothing representation in \eqref{eq:smoothing} can bias inference. For example, when $K>N$, many categories are expected to have zero counts, but the estimated probabilities for zero-count categories are strictly positive. Given the sparsity nature, an improperly chosen $\balpha$ automatically reduces the probabilities of seen events to give a non-zero probability to non-seen events. This yields poor results for datasets with sparse or quasi-sparse counts. A related concept is the membership distribution for topics in topic models. It is common to enforce sparsity in the mixing parameters to stabilize inference and enhance interpretability \citep{airoldi2016improving,xiao2014bayesian}.
 To address the sparsity pattern, some models employ two-component zero-inflated priors, a technique that has also been seen in Poisson \citep{lambert1992zero}, Negative binomial  \citep{yang2009testing}, DM \citep{koslovsky2023bayesian}, generalized DM \citep{tang2019zero} and logistic normal Multinomial \citep{zeng2022zero} distributions. With the two-component mixtures, category probabilities $\pi_k$ can be directly set to zero with a non-zero probability. However, they are insufficient flexibly for quasi-sparse data, and the two-component mixture form often leads to high computational costs.

In this paper, we propose a novel class of priors to address the two challenges described above. By utilizing hierarchical modeling, we transfer the reference prior to a `higher level' of the model, where we can impose desired structures on the model. First, our prior admits conjugate posterior computation for the DM distribution.  By utilizing partial fraction decomposition, we are able to provide closed-form posterior moments for homogeneous $\alpha$  An extension of the class also allows for marginal posterior inference on the concentration parameter, $\balpha$, enabling high dimensional marginalization of the class probabilities. Second, our default horseshoe type prior combined with a heterogeneous DM model provides continuous shrinkage, adapting to the degree of sparsity or quasi-sparsity in the data. This ``one-group'' approach to handling both zero and non-zero probabilities avoids the need for thresholding or two-component mixture, making it especially effective for data with small  nonzero counts.
 The shape of the prior is inspired by local-global shrinkage priors used in sparse Gaussian means and linear regression \citep{carvalho2010horseshoe,armagan2011generalized,armagan2013generalized,bhattacharya2015dirichlet}, as well as quasi-sparse count data \citep{datta2016bayesian}. Our prior is also related to the Stirling-Gamma prior proposed in the independent work of \citet{zito2023bayesian}.
Our prior places substantial mass around zero and possesses a heavy tail, allowing concentration parameters $\balpha$ to move freely between zero and infinity.

The rest of the paper is outlined as follows. 
\Cref{sec:prior_intro} provides the motivation for our priors and we introduce two classes of partial fraction priors  tailored for DM models. Additionally, we show how marginal posteriors and  posterior moments can be calculated exactly for homogeneous DM models. In \Cref{sec:posterior_inference}, we investigate how our default horseshoe prior  can achieve automatic adaptation to sparsity or quasi-sparsity in count datasets. In \Cref{sec:other_models}, we showcase how our priors offer a solution for other count models facing the same computational challenges, including the Negative Binomial and the  Generalized Dirichlet-Multinomial distribution.
We validate the efficacy of our method on simulated examples  in \Cref{sec:simulation} and illustrate its usage on real data analysis  in \Cref{sec:empirical}. Finally, \Cref{sec:discussion} wraps up with avenues for future research.

%

\vspace{-0.7cm}

\section{Conjugate Priors for Dirichlet-Multinomial Models} \label{sec:prior_intro}
\vspace{-0.5cm}

Our DM hierarchical model $\DM(\balpha, K, \bn)$ can be expressed as
\[\bn \mid \bpi \sim \Mu(\bpi, N), \quad \bpi\mid \balpha \sim \Dir(\balpha),\]
where the density function of multinomial distribution $\Mu(\bpi, N)$ is $
p(\bn \mid \bpi) = {N!}/{(n_1! \cdots n_K!)} \times \prod_{k=1}^K \pi_k^{n_k}.
$ and the Dirichlet density function is given by $p(\bpi | \balpha) = {\mB(\balpha)^{-1}} \prod_{k=1}^K \pi_k^{\alpha_k-1}$. Here $\mB(\balpha) = \prod_{k=1}^K \Gamma(\alpha_k)/\Gamma(A)$ with $A = \sum_{k=1}^K \alpha_k$ and $\Gamma(\cdot)$ is the Gamma function. The resulting posterior of $\bpi$ also follows a Dirichlet distribution,  given by
\begin{equation}\label{eq:multidir_post}
p(\bpi \mid  \bn, \balpha) = \frac{N!}{n_1! \cdots n_K!} \frac{1}{\mB(\balpha)} \prod_{k=1}^K \pi_k^{n_k+\alpha_k-1}\sim \Dir(\bn+\balpha).
\end{equation}
For simplicity of notation, we abuse the notation to denote the homogeneous Dirichlet distribution as $\Dir(\alpha)$, where $\alpha_k=\alpha$ for every $i= {1, \ldots, K}$,  and the corresponding DM distribution as $\DM(\alpha, K,\bn)$. 

In this section, we first describe  how we construct a prior class that facilitates conjugate updating of the concentration parameter $\balpha$. Exploiting the conjugacy, we provide closed-form representations for posterior mean of the homogeneous $\alpha$ and the corresponding category probabilities $\bpi$. Our results also provide a new avenue for inference for many other problems relying on homogeneous Dirichlet distribution.

\vspace{-0.6cm}
\subsection{Motivation} \label{sec:motivation}
\vspace{-0.3cm}

Our goal is to provide a flexible class of priors such that the marginal posterior inference for $\balpha$ is computationally straightforward. We begin by examining the posterior in \eqref{eq:multidir_post}.  One advantage of employing the Dirichlet prior is that we can marginalize out $ \bpi$ to obtain a marginal likelihood for $ \balpha $ given the counts as 
\[
p(\bn \mid \balpha) = \frac{N!}{n_1! \cdots n_K!} \frac{B(\bn+\balpha)}{\mB(\balpha)}
=\frac{N!}{\prod_k n_k!} \frac{1}{[A]^N }\prod_k [\alpha_k]^{n_k}  
\]
where $[x]^n:=x(x+1)\cdots(x+n-1)$ is the rising Pochhammer polynomial. Thus, for any prior $p(\balpha)$, by Bayes' rule, we obtain the marginal posterior as
\[
p( \balpha \mid \bn ) \propto p(\bn \mid \balpha)  p(\balpha) \propto  \frac{1}{[A]^N }\prod_k [\alpha_k]^{n_k}  \times p(\balpha).
\]

For simplicity, we first look into the homogeneous case where $\alpha_k=\alpha$, leading to $A=K\alpha$ and the posterior 
\begin{equation}\label{eq:alpha_cond_n}
p(\alpha \mid \bn)  =\frac{1}{C_N} \frac{\prod_{k=1}^K [\alpha]^{n_k}}{[K\alpha]^N} p( \alpha).
\end{equation}
where $C_N$ is a normalizing constant.

To construct a conjugate prior for the concentration parameter $\alpha$, we rewrite the fraction term as the sum of residuals (the explicit form of $\gamma_i$ will be introduced later) as 
\begin{equation}\label{eq:alpha_frac}
 \frac{\prod_{k=1}^K [\alpha]^{n_k}}{[K\alpha]^N}  = \sum_{i=1}^N \frac{\gamma_i \alpha}{ K \alpha + i -1}.
\end{equation}
Our objective now is to select a prior $p(\alpha)$ such that the following integral is in closed form
\begin{equation}\label{eq:ci}
c_i=\int_0^\infty \frac{\alpha p(\alpha)}{K\alpha+i-1}\d\alpha.
\end{equation}

Combining \eqref{eq:alpha_frac} and  \eqref{eq:ci}, the normalizing constant in \eqref{eq:alpha_cond_n}, also known as the marginal belief,  is 
\begin{equation}\label{eq:CN}
C_N=\int_{0}^\infty  \frac{\prod_{k=1}^K [\alpha]^{n_k}}{[K\alpha]^N} p( \alpha) \d\alpha =\sum_{i=1}^N \gamma_i c_i.
\end{equation}
Motivated by the residual representation of Pochhammer polynomial ratios in \eqref{eq:alpha_cond_n}, we introduce the Pochhammer prior family, such that \eqref{eq:alpha_frac} and  \eqref{eq:ci} can be satisfied, while also allows for straightforward computation of posterior moments and predictive inference.

\vspace{-0.6cm}

\subsection{Pochhammer Priors}

\vspace{-0.3cm}
We define the Pochhammer distribution as follows.

\begin{definition}[Pochhammer Distribution]\label{def:ph}
A random variable follows a Pochhammer distribution $\PH ( \alpha | m , a, b ,c ) $ if its density function satisfies the following form
\[
p(\alpha\mid  m , a , b , c ) \propto \frac{ [\alpha]^{m} }{ [ c \alpha  + a ]^{b}} \, \text{ with } \alpha \geq 0
\]
where $m, b$ are non-negative integers with $m \geq 0, b\geq m+2$, and $ a>0,  c>0$.
\end{definition}
We impose the condition $b\geq m+2$ to ensure the integrability of  the density function. To compute the normalizing constant for the Pochhammer distribution, we exploit the rational function nature of  the ratio and its partial fraction expansion representation.

\begin{theorem}\label{thm:ph_formula}
If a random variable $\alpha\sim \PH(\alpha \mid m,a,b,c)$, its density function can be written explicitly as
\[
p(\alpha\mid m, a,b,c)= C_{(m,a,b,c)}^{-1} \sum_{i=1}^{b} \frac{\gamma^{(m,a,b,c)}_i}{c \alpha + a + i -1 } 
\]
where 
\[ \gamma^{(m,a,b,c)}_i=\left\{ \begin{array}{cc}
 \frac{\prod_{s=1}^{m} ( 1+(s-1)c-a-i ) }{c^m \prod_{k=1, \; k \neq i }^{b} (k-i)}, & m\neq 0 \\
 \frac{1 }{\prod_{k=1, \; k \neq i }^{b} (k-i)}, & m=0
 \end{array}\right. \; \; C_{(m,a,b,c)}= \sum_{i=1}^b \frac{\big(-\gamma^{(m,a,b,c)}_i\big)}{c} \ln(a+i-1)\]
\end{theorem}

We provide the proof of this theorem in \Cref{proof:ph_formula}. The residual approach utilized here establishes a connection with the residual calculation techniques used in  P\'olya-Gamma data augmentation \citep{polson2013bayesian}.

\vspace{-0.6cm}

\subsection{Power Pochhammer Priors}
\vspace{-0.3cm}

The partial fraction class can be extended by size-biasing. We define the Power Pochhammer distribution as follows. 

\begin{definition}[Power Pochhammer Distribution]
 A random variable $\alpha$ follows a Power Pochhammer distribution $ \PPH ( \alpha \mid m , a, b ,c ,d ) $ if its  density function has the following form 
\begin{align*}
p( \alpha \mid m , a , b , c , d ) \propto \alpha^d \frac{  [\alpha]^{m} }{ [  c\alpha + a ]^{b} }  \; \; {\rm with} \; \; \alpha \geq 0
\end{align*}
where $b, d$ and $m$ are non-negative integers with $b\geq m+d+2$, and $a>0, c>0$.
\end{definition}

Apparently, the Pochhammer distribution is a special case of Power Pochhammer distribution when $d=1$. Similar to \Cref{thm:ph_formula}, we can derive the explicit form for the density function of  the Power Pochhammer distribution.

\begin{theorem}\label{thm:pph_formula}
If a random variable $\alpha \sim \PPH (\alpha\mid m, a,b,c,d)$, then its density function can be written explicitly as 
\begin{equation}
p(\alpha\mid m, a,b,c,d) = C_{(m,a,b,c,d)}^{-1} \sum_{i=1}^{b} \frac{\gamma^{(m,a,b,c,d)}_i}{c \alpha + a + i -1 } 
\end{equation}
where 
\[ \gamma^{(m,a,b,c,d)}_i= \left\{ \begin{array}{cc}  \frac{(1-a-i)^d\prod_{s=1}^{m} ( 1+(s-1)c-a-i ) }{c^{m+d} \prod_{k=1, \; k \neq i }^{b} (k-i)}, & m>0 \\
 \frac{(1-a-i)^d}{c^{d} \prod_{k=1, \; k \neq i }^{b} (k-i)}, & m=0 
 \end{array}\right. \; \; C_{(m,a,b,c,d)}= \sum_{i=1}^b \frac{(-\gamma^{(m,a,b,c,d)}_i)}{c} \ln(a+i-1)\]
\end{theorem}

Proof of this theorem is provided in \Cref{sec:proof_pph}.

\begin{corollary}[Moments of Power Pochhammer distributions]\label{coro:prior_moments} The distribution $\PPH(m,a,b,c,d)$ has up to $b-(m+d+2)$ moments. When $k\leq b-(m+d+2)$, \Cref{thm:pph_formula} can also be used to calculate the $k$-th moments of a  Power Pochhammer variable $\PPH(m,a,b,c,d)$ as
\[
\E (\alpha^k) = \int \alpha^k p(\alpha\mid m, a,b,c,d)\d\alpha  = \frac{C_{(m,a,b,c,d+k)}}{C_{(m,a,b,c,d)}} \int p(\alpha\mid m, a,b,c,d+k)\d\alpha =\frac{C_{(m,a,b,c,d+k)}}{C_{(m,a,b,c,d)}}
\] 
\end{corollary}

This identity also assists the calculation of posterior means and variances  for probabilities, $\bpi$, in the count model  via the law of iterated expectations, namely
$\E(\bpi\mid \bn)=\E_{\alpha\mid \bn} \E \big(\bpi \mid \alpha, \bn \big) $ where the inner term $ \E(\bpi \mid \alpha, \bn)$ is calculated as \eqref{eq:smoothing}.

\vspace{-0.6cm}

\subsection{Posterior Distribution for Homogeneous $\alpha$}
\vspace{-0.3cm}
Here we illustrate how our proposed Pochhammer prior can enable  closed-form representations for the posterior density and moments of a homogeneous DM model. We present one  such case below. The dependence of the residual terms $\gamma_i^*, \beta_i^*$ on $(\bn,m, a,b,c)$ is removed for notational simplicity.

\begin{theorem}\label{thm:ph_post2}
Under a Pochhammer prior $\alpha \sim \PH(m=n_0, a, b, c)$,  when there are no multiple roots in $[K\alpha]^n[c\alpha+a]^b$,  the posterior in \eqref{eq:alpha_cond_n} has a  closed-form representation as follows
{\footnotesize
\begin{equation}\label{eq:alpha_post}
p(\alpha \mid \bn)= C_{\bn} (m,a,b,c)^{-1}\frac{ \prod_{k=0}^K[\alpha]^{n_k}}{[K\alpha]^{N}[c\alpha+a]^b} =  C_{\bn} (m,a,b,c)^{-1} \left[\sum_{i=2}^{N}\frac{\gamma_i^*}{K\alpha+i-1}+\sum_{i=1}^b \frac{\beta_i^*}{c\alpha+a+i-1}\right]
\end{equation}}
where
\begin{align*}
\gamma_i^*&= \frac{\prod_{k=0}^K \prod_{s=1}^{n_k}(1+(s-1)K-i)}{K^{N+n_0-b} \prod_{s=1,\; s\neq i}^N (s-i)\prod_{t=1}^b (Ka+K(t-1)-c(i-1))}, \\
\beta_i^*&= \frac{\prod_{k=0}^K \prod_{s=1}^{n_k}(1+(s-1)c-a-i)}{c^{n_0} \prod_{s=1}^N \big(K+c(s-1)-K(a+i)\big)\prod_{t=1,\; t\neq i}^b (t-i)}, \\
C_{\bn} (m,a,b,c)&=\sum_{i=2}^N \frac{(-\gamma_i^*)}{K}\ln (i-1) +\sum_{i=1}^b  \frac{(-\beta_i^*)}{c}\ln (a+i-1). 
\end{align*}
\end{theorem}

\proof The calculation follows from the same residual argument in \Cref{thm:ph_formula}.

\begin{remark}[Multiple roots] We choose to illustrate the simplest case where no multiple roots exist in the denominator. 
We provide an example in \Cref{sec:double_roots} to show  how  the closed form representation can be derived when double roots are present. A more general recipe can be found in   \citet{zito2023bayesian}. 
\end{remark}

Similar to \Cref{coro:prior_moments}, the posterior means $\E(\pi_k\mid \bn)$ and $ \E \left ( \alpha \mid \bn \right ) $ can be calculated using a residue argument.

\begin{corollary}[Posterior Mean]\label{coro:post_mean_homo} When $b\geq m+2$, by the law of iterated expectation, we can calculate the posterior means of the category probabilities, ${\pi}_k, 1 \leq k \leq K $, via
\begin{align*}
 \E \left ( \pi_k \mid \bn \right )=   \E_{ \alpha \mid \bn } \E \left ( \pi_k | \alpha , \bn \right ) = \E_{ \alpha \mid \bn } \left ( \frac{n_k + \alpha}{ N + K \alpha } \right )
\end{align*}
Under a Pochhammer prior $\PH(n_0, a, b, c)$, we can write
\begin{align*}
\E_{ \alpha \mid \bn } \left ( \frac{n_k + \alpha}{ N + K \alpha } \right ) &=\frac{1}{C_{\bn}(n_0,a,b,c)} \int_{0}^\infty \left(\frac{n_k + \alpha}{ N + K \alpha }  \right) \frac{\prod_{i=0}^K [\alpha]^{n_k}}{[K\alpha]^{N}[c\alpha+a]^b} \; \d\alpha  =\frac{C_{\bn+e_k}(n_0,a,b,c)}{C_{\bn}(n_0,a,b,c)}
\end{align*}
where $e_i=(0,\ldots, 0,1,0,\ldots, 0)'$ is a vector that only the i-th element is 1.

When $b=m+2$, we can calculate the expectation of the category probabilities $\E(\pi_k\mid \bn)$ but the expectation $\E(\alpha\mid \bn)$ does not exist. When $b\geq m+3$, the posterior mean for the concentration parameter $\alpha$ under the prior $\PH(m,a,b,c)$ can be calculated by an extension of the Power Pochhammer class.
\end{corollary}

\begin{remark}[Importance of learning a single $\alpha$] While later we show that homogeneous DM models are less flexible compared to heterogeneous DM models in terms of adapting to sparsity, it remains an interesting problem to infer the value of  $\alpha$ in a homogeneous Dirichlet distribution from the data, such as the $\alpha$ in the Dirichlet-Laplace prior \citep{bhattacharya2015dirichlet}. More broadly speaking, there are many Bayesian models which involves a single parameter in a Gamma function or ratio of Gamma functions where we need posterior inference. Examples include mixtures of Dirichlet processes in \citet{escobar1995bayesian} and Ewens sampling Formula \citep{ewens1972sampling} (details in \Cref{sec:ESF}).
\end{remark}

\vspace{-0.7cm}

\section{Flexible Adaptation to Sparsity and Quasi-Sparsity}\label{sec:posterior_inference}

\vspace{-0.5cm}

In addition to the difficulty in inferring $\alpha$ due to the Gamma ratio structure, real applications often encounter discrepancies between prescribed model class and the empirical data distribution due to the presence of excess zeros and small nonzero counts. To flexibly accommodate inflated zeros, a common practice is to model the data distribution as a discrete two-component mixture, where one component is a Dirac mass at zero  and the other is the regular count model. For example, the zero-inflated DM models proposed by \citet{koslovsky2023bayesian} assigns a non-zero probability to events $\pi_k=0$ via an intermediate variable $c_k$ defined as
\begin{equation}\label{eq:ZIDM}
c_k \sim (1-\eta_k) \delta_0 (\cdot) +\eta_k \cdot \Ga(\gamma_k, 1), \text{ where } \eta_k \sim \text{Bernoulli}(\theta_k).
\end{equation}
Then $\pi_k$ is obtained as $\pi_k = {c_k}/{\sum_{j=1}^K c_j}$. While this method addresses the problem of excessive zeros, the model is insufficiently flexible for quasi-sparse datasets and the discrete two-component mixture structure is known to be computationally intensive.

In this section, we present a  \emph{``one-group answer"} to this two-groups problem of distinguishing between $\pi_k=0$ and $\pi_k\neq0$. The two key ingredients are heterogeneous DM models $\DM(\balpha, K, \bn)$ and a \emph{continuous} sparsity-inducing Pochhammer prior.  When the abundance of counts varies across different categories, it is more favorable to use heterogeneous DM models $\DM(\balpha, K, \bn)$, allowing for greater flexibility in adapting to the sparsity patterns in datasets. Regarding the prior, we show that with specific configurations, the prior exhibits two crucial properties essential for continuous shrinkage: (1) substantial mass around zero, and (2) a heavy tail.

%


\vspace{-0.6cm}

\subsection{Posterior Inference for Heterogeneous  $\alpha_k$}\label{sec:post_heter}
\vspace{-0.3cm}
Although our motivation for the Pochhammer prior originates from the fraction representation of homogeneous DM models, the Pochhammer prior can also be applied to heterogeneous DM models. We choose to place the same Pochhammer prior $\PH(\alpha_k\mid m, a, b, c)$ on each $\alpha_k, k=1, \ldots, K$, aiming to treat all $\pi_k$ and $\alpha_k$ equally when no external information is available.
Then the posterior in \eqref{eq:alpha_cond_n} is written as
\begin{equation}\label{eq:alpha_post_heter}
p(\balpha\mid \bn) \propto \frac{1}{[A]^N}  \prod_{k=1}^K \frac{[\alpha_k]^{n_k} [\alpha_k]^m}{[c\alpha_k+a]^b}
\end{equation}
where $A=\sum_{k=1}^K \alpha_k$. 

While the normalizing constant for the joint posterior in \eqref{eq:alpha_post_heter} is not straightforward to compute due to the term  $[A]^N$ in the denominator, $\{\alpha_k\}_{k=1}^K$ are independent of each other given $A$. Let $A_{-k}=A-\alpha_k$, we can write the conditional posterior distribution of $\alpha_k$ as
\begin{equation}\label{eq:ak_post}
p(\alpha_k\mid n_k, N, A_{-k}) \propto \frac{1}{[\alpha_k + A_{-k}]^N}\frac{[\alpha_k]^{n_k}[\alpha_k]^{m}}{[c\alpha_k+a]^b}.
\end{equation}

Similar to \Cref{thm:ph_post2}, we can provide a closed-form representation for the conditional density in \eqref{eq:ak_post}, with details deferred to  \Cref{sec:heter_cond}. Moreover, the corresponding conditional posterior mean for the category probabilities $\pi_k$ can be written as
\begin{align*}
\E[\pi_k\mid n_k, N, A_{-k}] = \frac{C_{n_k+1, N+1, A_{-k}}(m,a,b,c)}{C_{n_k, N, A_{-k}}(m,a,b,c)}
\end{align*}

We employ Metropolis-Within-Gibbs sampling to obtain draws from the posterior in \eqref{eq:ak_post}, which we  outline  in \Cref{sec:heter_cond}. If feasible, direct sampling from the Gibbs conditional distribution \eqref{eq:ak_post}  could streamline our approach and enhance  computation efficiency. However, we hope to explore such implementation of the Pochhammer Gibbs conditional in future research.

\vspace{-0.6cm}
\subsection{Horseshoe Priors}
\vspace{-0.3cm}
Before we delve into how our method can flexibly adapt to sparsity, we want to examine the how the shape of the Pochhammer distribution changes with respect to the change in each parameter. From the definition of the Pochhammer distribution, parameters $m$ and $b$ significantly influence the shape of the distribution, determining the available moments. Specifically, the prior will have a non-diminishing mass around $0$ only if $m=0$, while $b$ controls the heaviness of the right tail.

\citet{polson2010shrink} suggest that a desired continuous shrinkage prior should  place a non-decaying mass around zero and possess a heavy tail. We observe that the configuration of $m=0, b=2$ can provide us such ideal properties. Under this configuration,  the density function can  be rewritten as
\[
p(\alpha\mid m=0, a, b=2,c)
\propto \bigg( \frac{1}{\alpha+\frac{a}{c}}-\frac{1}{\alpha+\frac{a+1}{c}}\bigg).
\]
When $a$ is small and $c$ is large, 
i.e. both $\frac{a}{c}$ and $\frac{a+1}{c}$ are very small, this is the ``closest to non-integrable''  (on $\R^+$) fraction prior, exhibiting behavior similar to a horseshoe prior \citep{carvalho2010horseshoe}. It concentrates mass near $\alpha=0$, while its fat right tail allows the posterior to explore large values as well. 
As the distribution is supported on $\R^+$, we refer to the Pochhammer prior with $m=0$ as the ``Half-Horseshoe" prior.

\begin{figure}[!ht]
\centering
\subfigure[$m$]{
\includegraphics[width=0.4\textwidth]{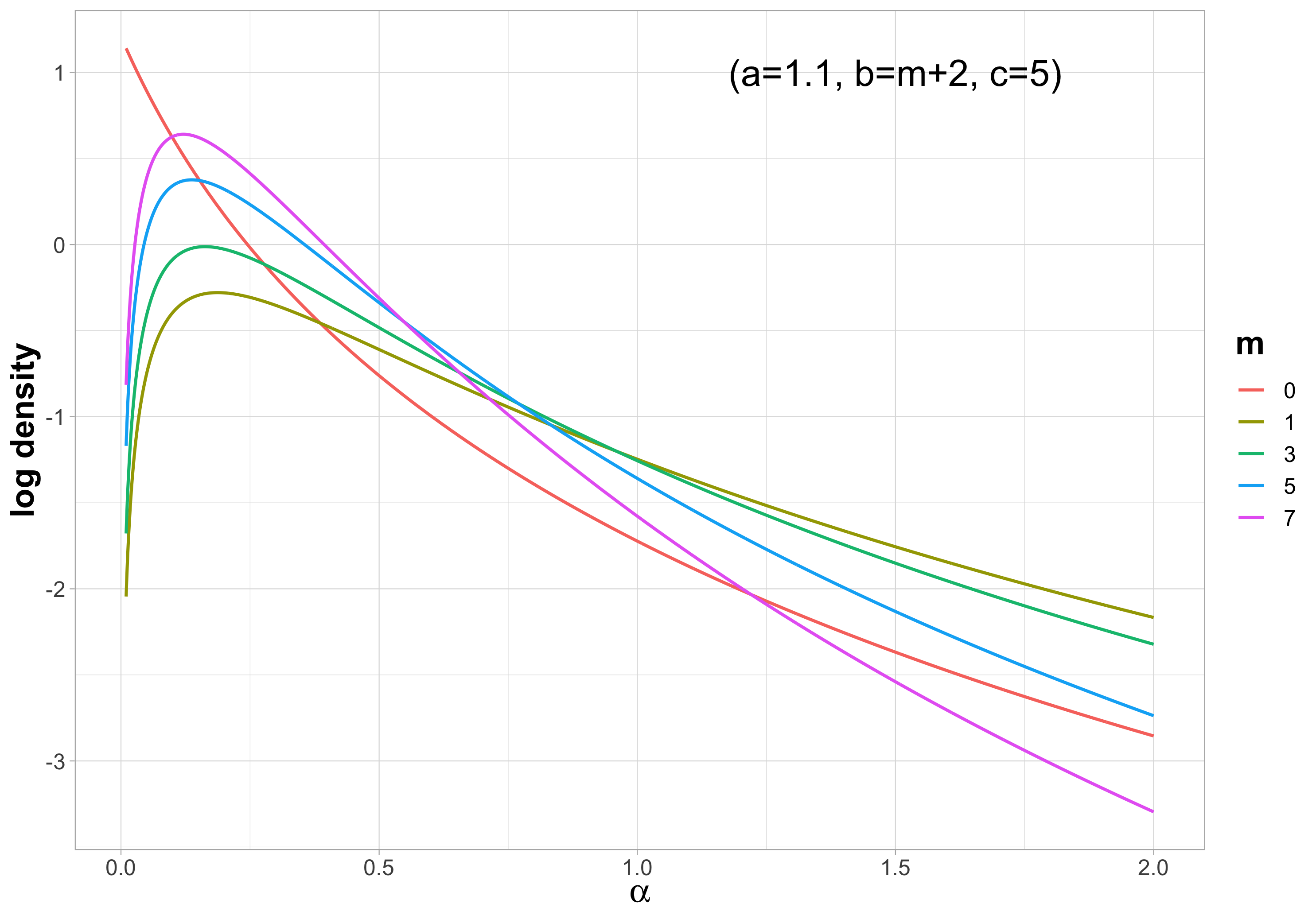}
}
\subfigure[$a$]{
\includegraphics[width=0.4\textwidth]{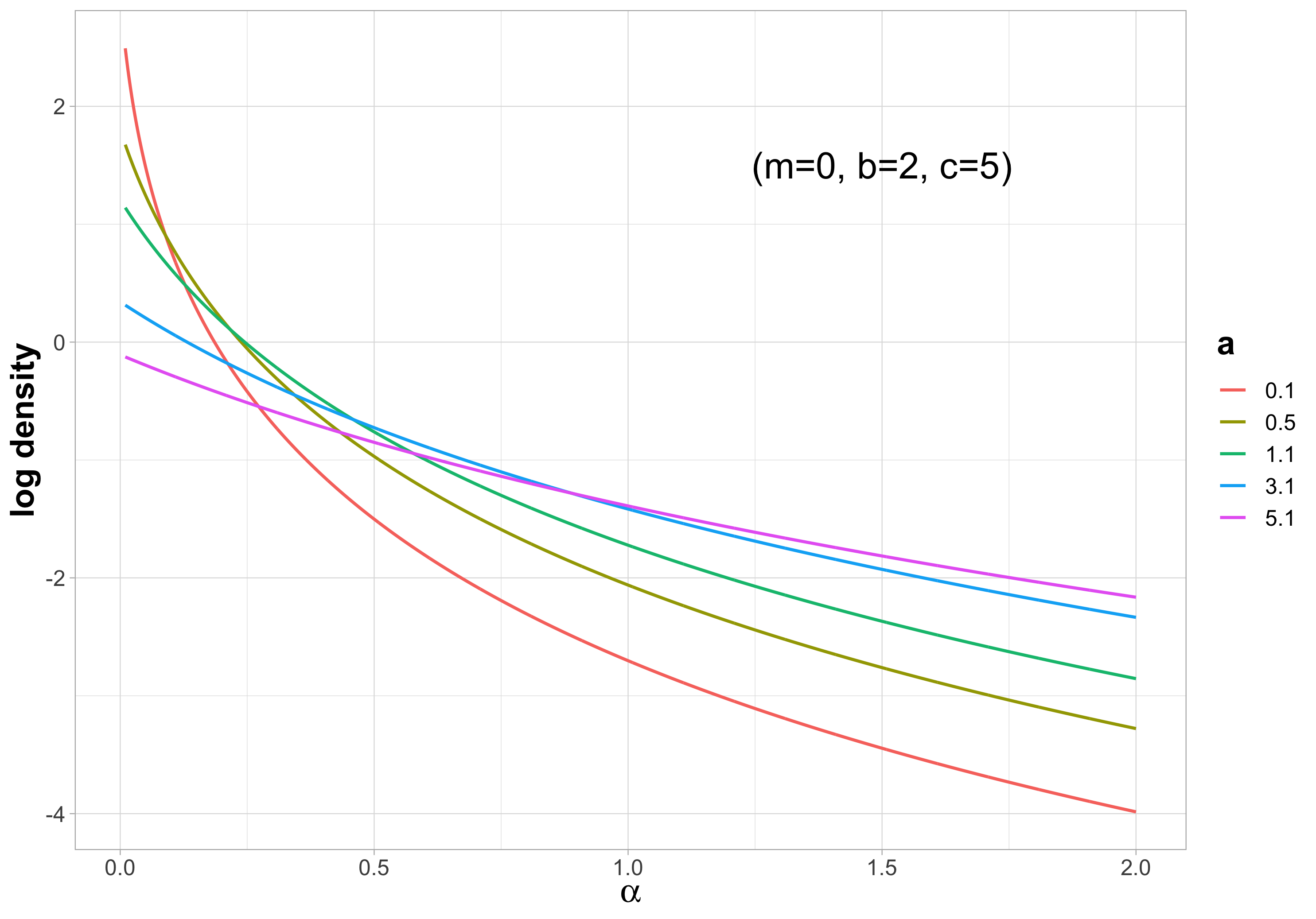}
}
\subfigure[$b$]{
\includegraphics[width=0.4\textwidth]{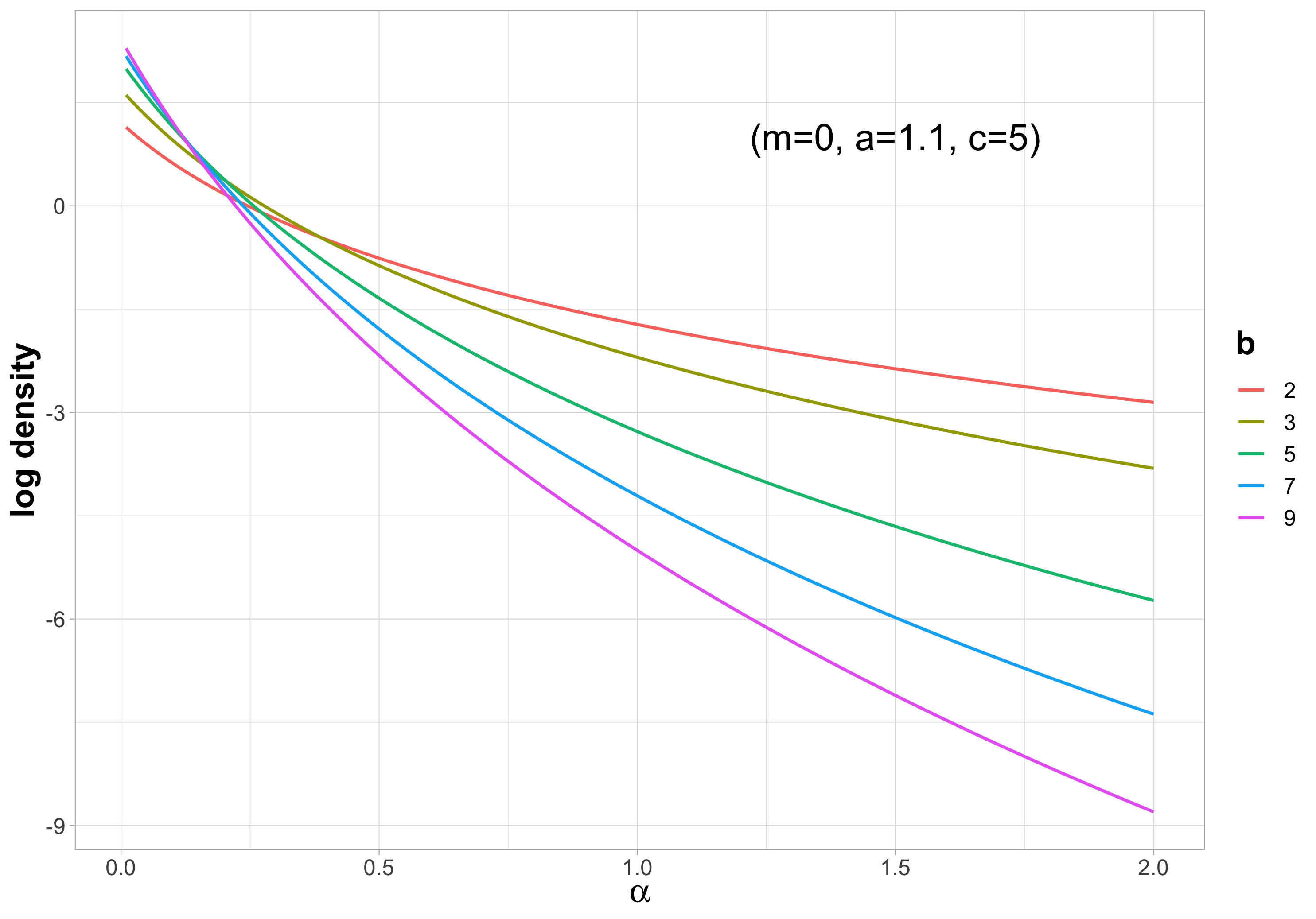}
}
\subfigure[$c$]{
\includegraphics[width=0.4\textwidth]{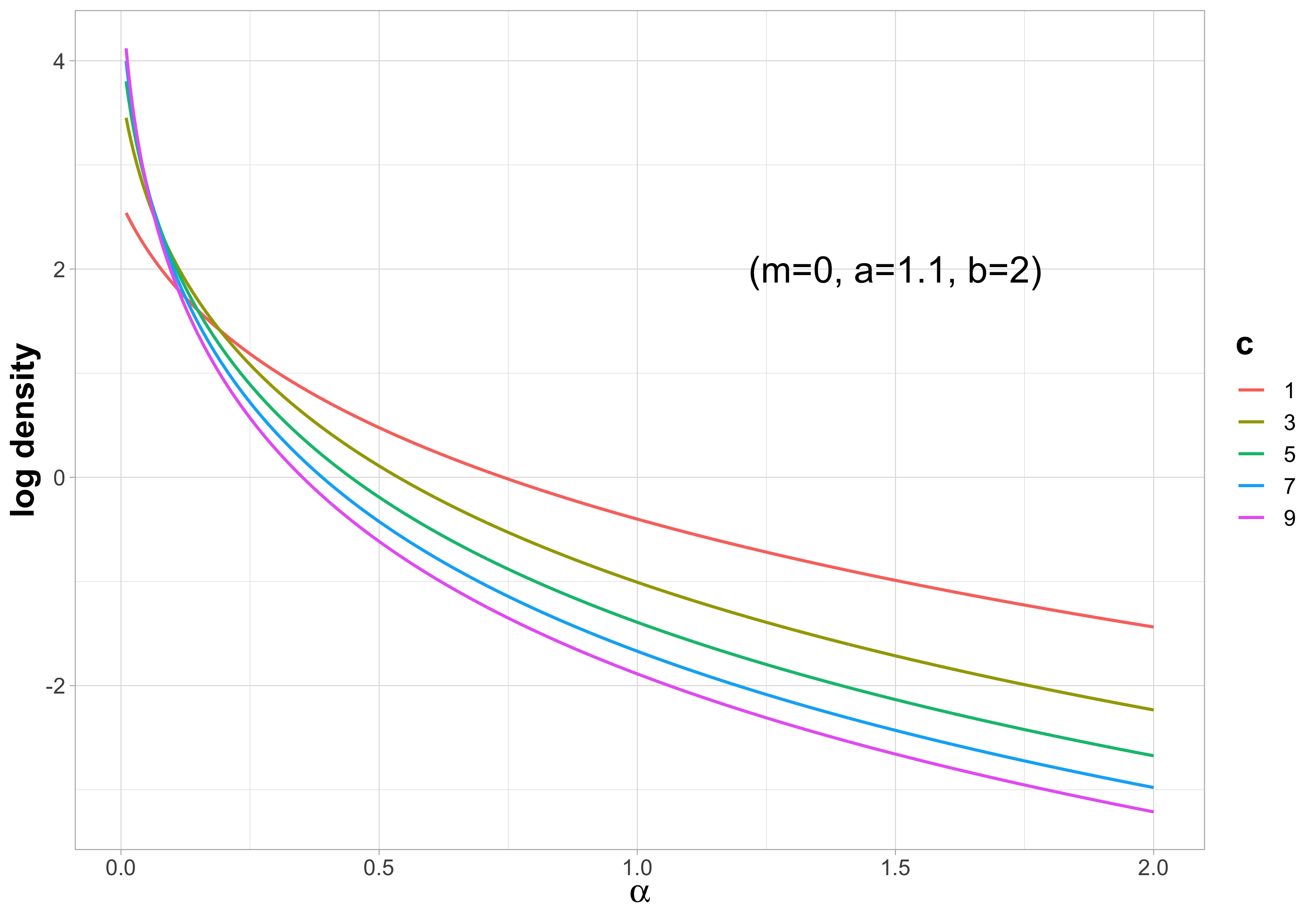}
}
\caption{\small Impact of hyperparameters on the shape of the Pochhammer distribution, shown in log scale. The baseline choice is $(m=0, a=1.1, c=5)$.}\label{fig:ph_shape}
\vspace{-0.5cm}
\end{figure}

We provide intuitions on choosing other hyperparameters for the Pochhammer prior in \Cref{fig:ph_shape}. The baseline choice is $(m=0, a=1.1, b=2, c=5)$, and we examine the sensitivity of the prior's shape to changes in each parameter. 
When $m=0$, the tail is lighter with decreasing $a$,  increasing $b$, or  increasing $c$. Parameter $b$ determines the number of moments the prior and posterior distributions  have. It requires $b\geq m+2$ for the density to be integrable on $\R^+$, and both the prior and  posterior distributions have up to $(b-m-2)$ moments. We find that our posterior performance in the simulated examples in \Cref{sec:simulation} is not very sensitive to the choices of $a$ and $c$, therefore, we proceed with $a=1$ and $c=1$.

Next, we verify the theoretical properties of our prior: heavy-tailed (\Cref{prop:heavy_tail}) and \emph{super-efficient} sparsity recovery  (\Cref{prop:KL_risk}).

\begin{proposition}[Heavy Tail]\label{prop:heavy_tail}The Pochhammer distribution $\alpha\sim \PH(m,a,b,c)$ exhibits heavy-tailed behavior, namely
\[
\lim_{x\to\infty} e^{tx} P(\alpha>x)=\infty \quad \text{ for all } t>0.
\]
\end{proposition}

The proof is provided in \Cref{proof:heavy_tail}.

\begin{proposition}[Kullback-Leibler Risk Bounds from \citet{polson2010shrink}]\label{prop:KL_risk} Let $A_\epsilon = \{\balpha: L(p_{\balpha_0}, p_{\balpha})\leq \epsilon\} \in \R$ denote the Kullback-Leibler information neighborhood of size $\epsilon$, centered at $\balpha_0$. Let $\mu_N(\d\balpha)$ be the posterior distribution under $p(\balpha)$ after observing the N counts data $\bn=(n_1, n_2, \ldots, n_K)'$, and let $\hat p_N =\int p_\balpha \mu_N (\d\balpha)$ be the posterior mean estimator of the density function.

Suppose that the prior $p(\balpha)$ places positive mass in the neighborhood around $p_{\balpha_0}$, i.e., $\mu(A_\epsilon)>0$ for all $\epsilon>0$. Then the following bound for $R_N$, the Ces\`aro-average risk of the Bayes estimator $\hat p_N$, holds for all $\epsilon>0$:
\[
R_N =\frac{1}{N}\sum_{j=1}^N L(p_{\balpha_0} , \hat p_j)\leq \epsilon- \frac{1}{N} \log \int_{\balpha-\sqrt\epsilon}^{\balpha-\sqrt\epsilon} p(\balpha)\d \balpha.
\]
\end{proposition}
Intuitively the more mass the prior $p(\balpha)$ has in a neighborhood near the true value $\balpha_0$, the better the bound is. Specifically, for the special case where $\alpha_{k}=0$ for some $k$, the risk bound can be greatly improved if the prior density has a pole at zero.

For our half-horseshoe Pochhammer prior,  $\PH(m=0, a=1,b=2,c=1)$, we have
\begin{align*}
\int_0^{\sqrt \epsilon} p(\alpha)\d\alpha & = \frac{1}{\log 2} \int_0^{\sqrt \epsilon} \frac{1}{(\alpha+1)(\alpha+2)}\d\alpha = \frac{1}{\log 2} \Big[ \log(\alpha+1)  {\Big|}_0^{\sqrt \epsilon}  -  \log(\alpha+2) {\Big|}_0^{\sqrt \epsilon} \Big] \\
& = \frac{1}{\log 2}  \log(1+\frac{\sqrt \epsilon}{2+\sqrt \epsilon}) \approx  \frac{1}{\log 2}  \times \frac{\sqrt \epsilon}{2+\sqrt \epsilon} = o(\sqrt \epsilon)  \text{ when } \epsilon \to 0.
\end{align*}
Thus the prior mass in the neighborhood is of order $\sqrt \epsilon$, which enables fast recovery of the true sampling distribution in sparse situations. This is one example of a KL ``\emph{super efficient}'' prior.

In addition to the two properties above, we show that the limiting distribution when $b\to \infty$ is the exponential distribution $\Exp(1)$.

\begin{proposition}[Limiting Distribution When $b\to \infty$]
Let $\alpha\sim \PH(m=0, a=1,b, c=1)$, then the following convergence in distribution holds:
\[
\alpha\log b \to \gamma,\qquad  \gamma \sim \Exp(1), \quad b\to \infty
\]
\end{proposition}
\begin{proof}
Plugging in the parameters we have
\[
p(\alpha)\propto\frac{1}{[\alpha+1]^b}=\frac{\alpha}{[\alpha]^{b+1}},
\]
which is equivalent to a Stirling-Gamma $\Sg(2,1,b)$ distribution. The rest of the proof follows from Proposition 1 of  \citet{zito2023bayesian}.
\end{proof}

The proposition also implies that $\alpha\to 0$ in probability as $b\to \infty$ with a logarithmic rate of convergence via Slutzky's theorem. We include a comparison between two distributions $\PH(m=0,a=1,b,c=1)$ and ${\rm Ga}(1, \log b)$ in \Cref{fig:gamma_comp}. While the two distributions look more similar as $b$ increases, it is worth noting that the Pochhammer distribution is  heavy-tailed, with a larger mass near zero and a heavier right tail compared to the Gamma distribution.

\begin{figure}[!ht]
\centering
\includegraphics[width=0.5\textwidth]{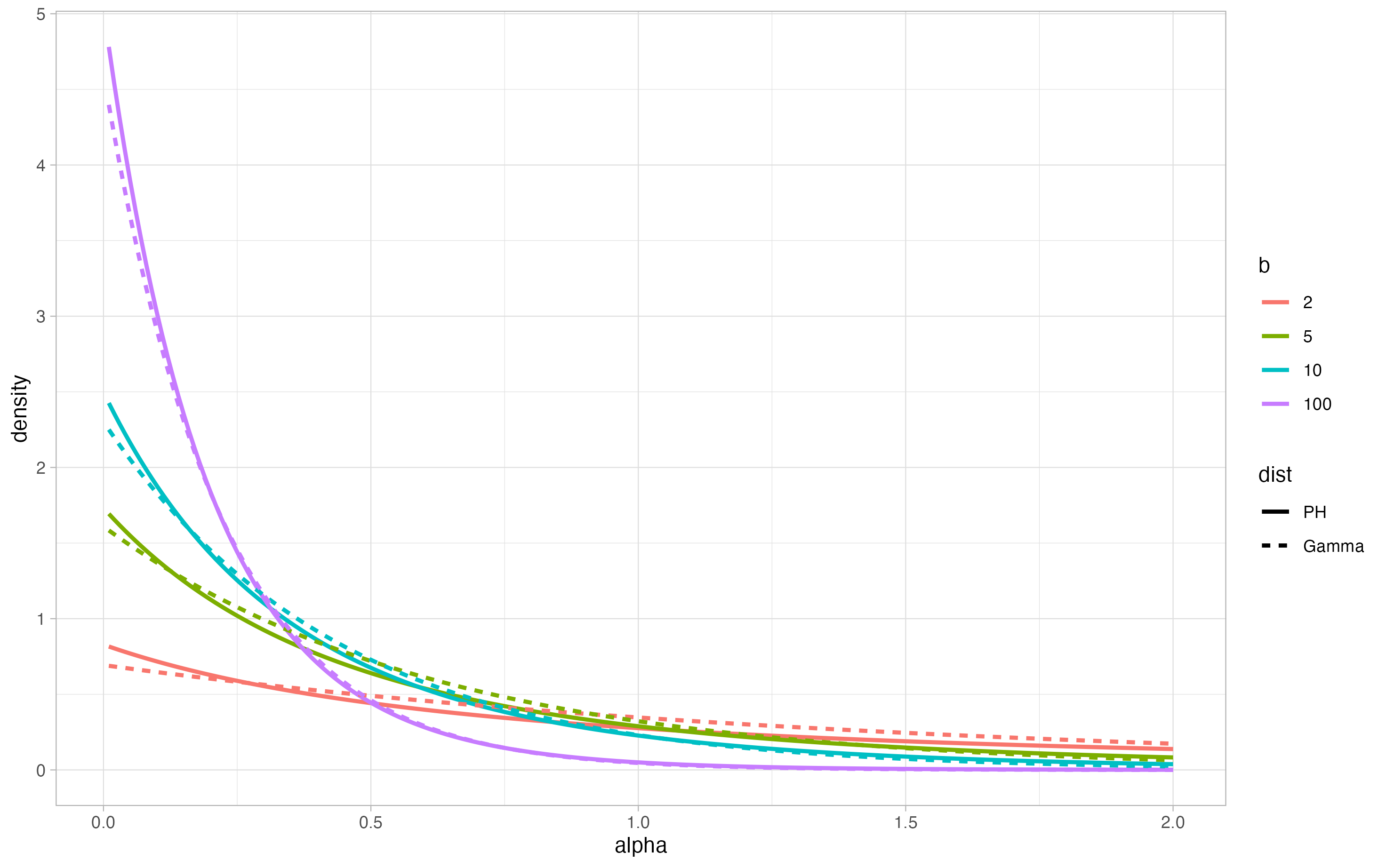}
\caption{Comparison of density functions of $\PH(m=0,a=1, b, c=1)$ distribution and ${\rm Ga}(1, \log b)$ distribution.}\label{fig:gamma_comp}
\end{figure}


\vspace{-0.5cm}

\begin{remark}[Connection to \citet{zito2023bayesian}] The Stirling-Gamma distributions proposed by \citet{zito2023bayesian} have a similar form as our Pochhammer distributions and serve as a conjugate prior for the Dirichlet process. Despite the similar formula, our motivation and application are drastically different from theirs. We aim to propose a continuous shrinkage prior for sparse or quasi-sparse count datasets, while \citet{zito2023bayesian} focus on learning the number of clusters in the mixture models. Furthermore, the numerator of our marginal likelihood $p(\bn\mid \balpha)$ contains a product of Pochhammer polynomials, while the marginal likelihood of partitions $p(\Pi_n =\{C_1, \ldots, C_k\}\mid \alpha) $ has powers of $\alpha$ in the numerator, which makes their posterior analysis more accessible than ours.
\end{remark}

\vspace{-0.6cm}

\subsection{Automatic Adaptation to Sparsity}
\vspace{-0.3cm}
A major challenge in modeling high-dimensional compositional count datasets is that the DM models cannot accommodate zero counts or counts close to zero very well. While Zero-Inflated DM (ZIDM) models \citep{koslovsky2023bayesian} assign a non-zero probability to events $\pi_k=0$ to fix the issue, it will not be able to properly handle the quasi-sparse case in which many of the probabilities are close to zero but not zero. Inspired by the continuous shrinkage idea used in horseshoe priors \citep{carvalho2010horseshoe}, we recommend the configuration of $m=0, b=2$ as our default prior, which possesses a non-decaying mass around zero and a heavy tail. 

With this configuration combined with heterogeneous DM distribution, each concentration parameter $\alpha_k$ can freely move between $0$ and $\infty$. When the posterior density of $\alpha_k$ has a pole at zero, the corresponding posterior of $\pi_k $ will also place non-zero mass near zero. Compared to ZIDM models, we can better monitor quasi-sparse probability vectors. 

\begin{remark}[Shrinking Posteriors For Zero Counts] Under a shrinkage prior $\PH(m=0, a, b, c)$, if the count $n_k=0$, the posterior in \eqref{eq:ak_post} can be simplified as
\[
p(\alpha_k\mid n_k=0, N, A_{-k})  \propto  \frac{1}{[\alpha_k + A_{-k}]^N}\frac{1}{[c\alpha_k+a]^b}.
\]
Given $A_{-k}>0$, the posterior distribution of $\alpha_k$ retains a horseshoe shape again,
which is desirable for analyzing sparse count datasets.
\end{remark}

\vspace{-0.6cm}
\section{Using Pochhammer Priors for Other Count Models}\label{sec:other_models}
\vspace{-0.3cm}
Choosing the appropriate concentration parameter $\alpha$ and properly accommodating the sparsity pattern are common challenges across many count models. Leveraging the shared gamma ratio representation, our Pochhammer priors can serve as a convenient solution for a plethora of count models. In this section, we explore how PH priors can facilitate inference  for Negative Binomial and Generalized Dirichlet-Multinomial Models. We provide two more examples, the Yule-Simon distribution and Ewen's sampling formula, in \Cref{sec:other_models_p2}.

\vspace{-0.6cm}
\subsection{Negative Binomial Models}
\vspace{-0.3cm}
The Negative Binomial (NB) model is widely used to model overdispersed count data, where the variance exceeds the mean. Originally introduced by \citet{greenwood1920inquiry} for modeling accident statistics, it has since found applications across many fields, including  infectious disease modeling \citep{lloyd2007maximum}, species abundance in ecology \citep{linden2011using}, and customer purchasing behavior. 

 Despite its widespread use, similar to DM models,  NB models  suffer from two computational issues. First, their hyperparameters $(\alpha, \pi) $ (definition given later) are often treated as given or being estimated via empirical Bayes method, which fails to account for uncertainty  and limits the ability to  incorporate  any prior knowledge. While several researchers have attempted to use Bayesian inference for NB models by specifying a prior distribution for $\alpha$ and $\pi$,  none  have achieved \emph{closed-form inference}, with maybe \citet{bradlow2002bayesian} being the closest using polynomial expansions. The second challenge relates to the widely observed excessive zeros in count datasets. A common choice is the Zero-Inflated Negative Binomial (ZINB) model \citep{lambert1992zero}, which models the data as a mixture of negative binomial distribution and a spike at zeros. Due to its discrete nature, ZINB does not scale well with high dimensional data and many have proposed methods to improve its computational efficiency (e.g. \citet{neelon2019bayesian}).

We now present how PH priors can facilitate conjugate inference for NB models. First, we review the definition of the NB model. Under a $\NB(\alpha, \pi)$ model, the probability observing a individual count  $n_i$  is 
\[
p(n_i \mid \alpha, \pi) =\frac{\Gamma(n_i+\alpha)}{\Gamma(\alpha) n_i!} \pi^\alpha (1-\pi)^{n_i} =\frac{[\alpha]^{n_i} }{n_i!} \pi^\alpha (1-\pi)^{n_i}.
\]
The NB model can also be viewed as a continuous mixture of Poisson distributions, where each count $n_i$ is distributed $\rm{Poisson}(\lambda_i)$ and the rate parameter $\lambda_i$ follows a Gamma distribution $\Ga(\alpha, \frac{\pi}{1-\pi})$.

As \citet{bradlow2002bayesian} point out, the main challenge in achieving a conjugate or closed-form solution centers around $\alpha$, while one can easily specify a prior for $\pi$ (or some function of $\pi$) to have it integrated out. To address this issue, we introduce prior dependence  in the conditional $p(\pi\mid \alpha)$, similar to the structure in Normal-inverse-Wishart prior for multivariate Gaussian distribution.  Since the two parameters $\alpha$ and $\pi$ are intertwined in the likelihood function, we model their dependence via a beta distribution to enable joint conjugacy as
\[
\pi \mid \alpha \sim \Be ( c\alpha+a, b)  \text{ and } \alpha \sim \PH(m,a,b,c).
\]
The resulting joint prior has the form
\begin{align*}
p(\pi, \alpha) = p(\pi\mid \alpha)p(\alpha) 
\propto [\alpha]^m \pi^{c\alpha+a-1}(1-\pi)^{b-1}.
\end{align*}

Supposed we observe a vector of K counts $\bn= (n_1, \ldots, n_k)$ which are i.i.d. sampled from a negative binomial distribution $\NB(\alpha, \pi)$, the joint posterior is 
\begin{align*}
p(\pi, \alpha\mid \bn) &= \left\{\prod_{k=1}^K p(n_k\mid \alpha, \pi) \right\}p(\pi\mid\alpha)p(\alpha)\propto\pi^{(c+K)\alpha+a-1}(1-\pi)^{N+b-1} \prod_{k=0}^K [\alpha]^{n_k}
\end{align*}
where $N=\sum_{k=1}^K n_k$. Then similar to DM models, we  marginalize out $\pi$  and get
\begin{align*}
p(\alpha\mid \bn) = \int p(\pi, \alpha\mid \bn)\d \pi
\propto \frac{ \prod_{k=0}^K [\alpha]^{n_k}}{[(c+K)\alpha+a]^{N+b}}.
\end{align*} 

Again, the normalizing constant is available in closed-form, which can be derived using the same partial fraction decomposition argument.  The posterior moments $\E(\alpha^k\mid \bn)$ can also be calculated in closed-form with a similar calculation from \Cref{coro:post_mean_homo}. The conditional posterior $p(\pi\mid \alpha, \bn)$ can be sampled from a beta distribution.

Taking a closer look at the marginal likelihood $p(n_i\mid\alpha)$ under our coupled prior 
\begin{align*}
p(n_i\mid \alpha)=\int_0^\infty p(n_i\mid \alpha,\pi)p(\pi\mid\alpha)\d\pi, 
\end{align*}
we essentially reconstruct the model as a continuous mixture of NB distributions depending only on $\alpha$, where the mixing distribution is a Beta distribution. Although this might come out as counter-intuitive, we now show this Negative Binomial-Beta mixture grants us more flexibility in modeling count datasets with different sparsity pattern. 

\begin{figure}[!ht]
\centering
\includegraphics[width=0.9\textwidth]{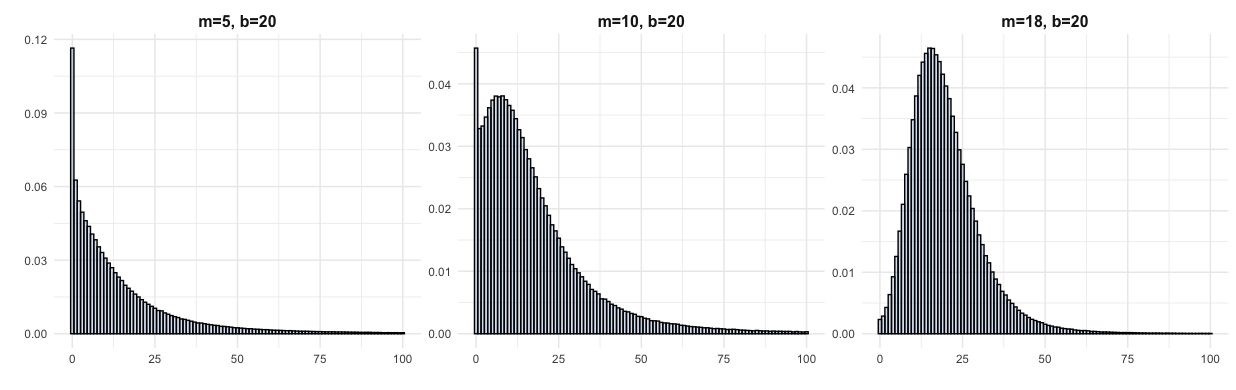}
\vspace{-0.5cm}
\caption{Density plots of 100\,000 $n_i$'s generated from $\NB(\alpha, \pi)$ with our coupled prior. }\label{fig:NB_plot}
\vspace{-0.5cm}
\end{figure} 

We provide a visualization of the density of $\bn$ under different specification of PH priors with fixed $a=c=1$ in \Cref{fig:NB_plot}.  We observe an interesting change in the sparsity pattern in the counts. With different degree difference $b-m$, the shape of the distribution goes from a significant spike at zero and sharply decreasing density, to a mixture of a spike at zero and a regular NB distribution, to  eventually similar to a regular NB distribution. These three plots also corresponds to (1) quasi-sparse counts that has an overabundance of zeros and small non-zero counts; (2) zero-inflated counts; (3) overdispersed counts with no zero inflation.
Overall, we see that parameters $m$ and $b$ control the size of the spike of zeros and the mode of the nonzero part. Parameters $a$ and $c$ are less influential to the shape so we just fix their values. 

Our findings in \Cref{fig:NB_plot} suggests that one can prefer to use the NB-Beta mixture coupled with our PH priors to model sparse counts similar to the two-component mixture ZINB models. Our construction admits conjugate updating and closed-form representation of posterior moments while still properly accommodating excessive zeros. 
\vspace{-0.6cm}
\subsection{Generalized Dirichlet-Multinomial distribution} \label{sec:GDM}
\vspace{-0.3cm}
One limitation of applying DM models in real-world application is that they imply negative correlation between counts, whereas actual data can display both negative and positive correlations.  For example, in RNA-seq data analysis,  two exon sets may belong to one or a few RNA isoforms, leading to complicated correlation structures among their counts \citep{zhang2017regression}. For data with known positive pairwise correlation structure, one would choose Negative Multinomial (NM) models over DM. The Generalized Dirichlet-Multinomial (GDM) model, 
on the other hand, offers greatest flexibility by accommodating general correlation structures. 

Once again, the two computation challenges, stemming from  the gamma ratio expression and excessive zeros and small counts, persist for both NM and GDM. \citet{tang2019zero} propose the zero-inflated GDM to handle excessive zeros  in taxon counts and the complex correlation structure and dispersion patterns among taxa. Now we show how our PH priors can assist with NM and GDM models.

The NM model is a generalization of NB model to more than two outcomes. Using the same notation as in the DM model, the probability  for observing count vector $\bn =(n_1, n_2, \ldots, n_K)$ under  $\NM(\alpha, \bpi)$ is given by
\begin{equation}
p(\bn \mid \alpha, \bpi) = \Gamma(N+\alpha) \frac{\pi_0^\alpha}{\Gamma(\alpha)} \prod_{k=1}^K \frac{\pi_k^{n_k}}{n_k!} \propto [\alpha]^N \pi_0^\alpha \prod_{k=1}^K \pi_k^{n_k}
\end{equation}

Thus, the marginal likelihood $p(\bn\mid \alpha)$ can be written as
\[
p(\bn\mid\alpha) = \int p(\bn \mid \alpha, \bpi)\d \pi  \propto \frac{\alpha}{[\alpha+N]^{K+1}}.
\]
By assigning a PH prior to $\alpha$, we immediately obtain closed-form representations for the posterior and its moments. If one wishes to use a Dirichlet prior on $\bpi$, conjugate conditional posterior similar to \eqref{eq:alpha_post_heter} can be derived. Additional shrinkage on $\pi_k$ can be encouraged with a half-horseshoe prior on the Dirichlet distribution.

For the GDM model, we adopt the stick-breaking representation. The GDM assumes the count vector $\bn =(n_1, n_2, \ldots, n_K)'$ follows a multinomial distribution $\Mu(\pi, N)$, where the category probabilities vector $\pi$ is constructed from a set of mutually independent Beta variables $ \mathbf{Z}= (Z_1, Z_2, \ldots, Z_K)$ as
\[
\pi_1= Z_1, \, \pi_j = Z_j \prod_{k=1}^{j-1} (1-Z_k), j=2, \ldots, K 
\]
with each $Z_k$ follows a Beta distribution $(\alpha_k, \beta_k)$. 

If $\beta_{k-1}= \alpha_k+\beta_k$ for $2\leq k\leq K$, then the GDM reduces to the standard DM. The extra parameters $\beta_k$ allow for more flexibility in modeling heterogeneous dispersion levels. In addition, $\alpha_k$'s and $\beta_k$'s jointly determine the correlation structure in the data. 

Note that $\pi_k \to 0$ when $Z_k \to 0$. To accommodate excessive zeros in the dataset, \citet{tang2019zero}  model $Z_k$ as a mixture of a Dirac mass at zero and $\Be(\alpha_k, \beta_k)$. Our solution is to impose the half-horseshoe PH prior $\PH(m=0, a=1, b=2, c=1)$ on $\alpha_k$'s to encourage shrinkage on $\pi_k$. The shape of the Beta distribution with a half-horseshoe prior on $\alpha$ is illustrated in \Cref{fig:GDM_density}. We observe that when  $\alpha$ and $\beta$ both follow the half-horseshoe distribution, the latent variable $Z$ has a horseshoe-shaped density, encouraging values near either 0 or 1. For GDM, when $Z_j\to 1$, it shrinks the sum of $\{\pi_k\}_{k=j+1}^K$ to zero, which is not desired. Thus, we suggest one can use a non-shrinking PH prior on $\beta_k$  to enable conjugate updating for all parameters, or update $\beta_k$ via EM to reduce the computation costs.

\begin{figure}
\centering
\includegraphics[width=0.5\textwidth]{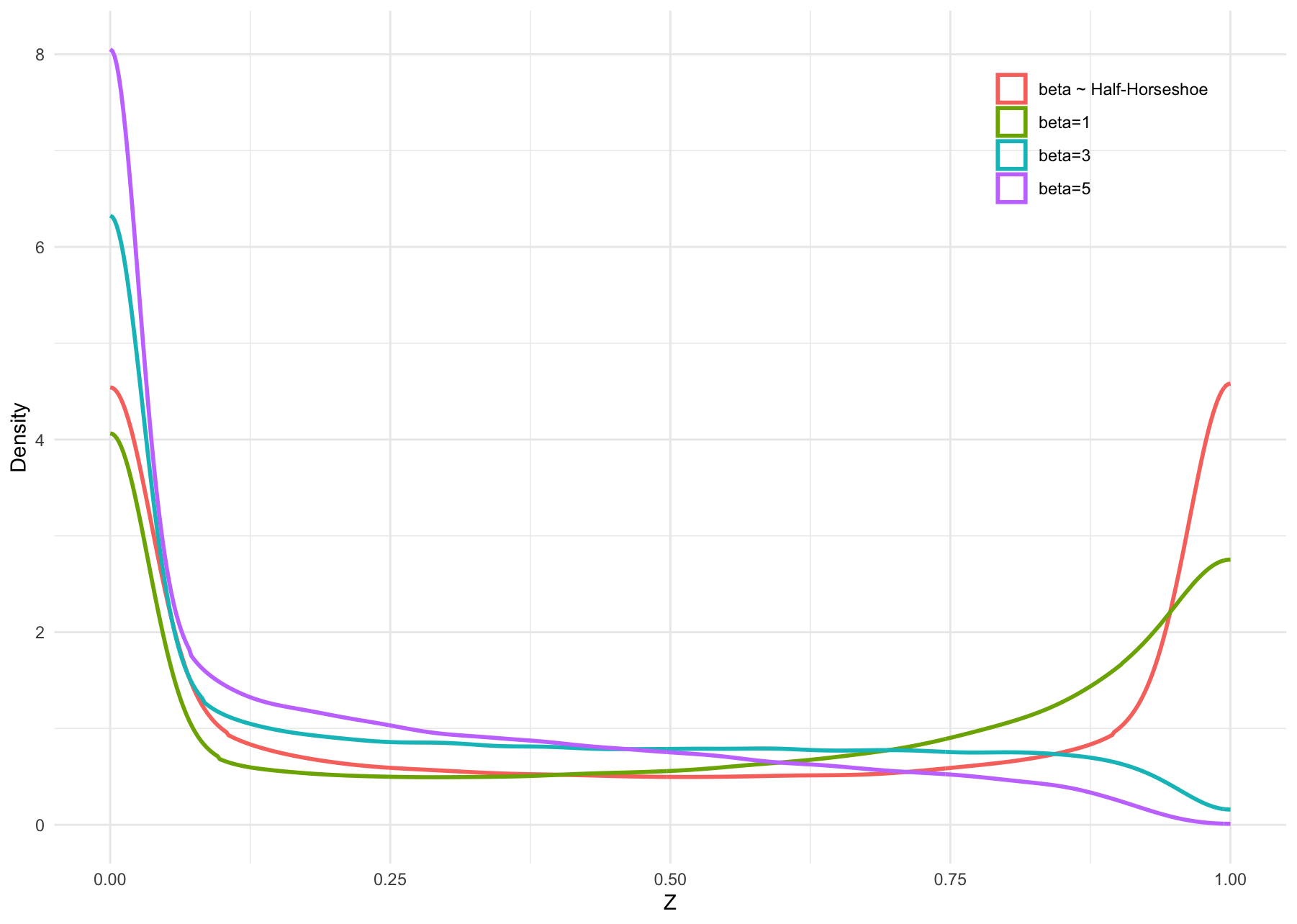}
\caption{Density Plot of  $Z \sim \Be(\alpha, \beta)$ distribution. Here $\alpha$ follows the Half-Horseshoe prior, and we try $\beta$ set at fixed values, or sample from a Half-Horseshoe Prior.}\label{fig:GDM_density}
\end{figure}

\vspace{-0.6cm}
\section{Simulations}\label{sec:simulation}
\vspace{-0.3cm}
We compare our methods with the following approaches: (1) A Bayesian DM model; (2) Tuyl's approach \citep{tuyl2018method}; (3) A zero-inflated DM (ZIDM) model by \citet{koslovsky2023bayesian}, where a latent variable is used to allow for exact zero probabilities. For each method, we obtain posterior draws from 10\,000 MCMC iterations. We report the average absolute value of the difference between the estimated and true probabilities $\text{ABS}(\pi)=n^{-1}\E\norm{\hat \pi -\pi}_1$ and 95\% coverage probabilities (COV). We compute the statistics from 20 replicated datasets for each setting.

For our method, we consider both homogeneous (PH-h)  and heterogeneous (PH-d) versions. To investigate  how sensitive  the posterior is to the shape of the prior distribution, we  fix $a=c=1$ and include 4 configurations of parameters: (1) $m=0, b=2$; (2) $m=0, b=5$; (3) $m=1, b=3$; (4) $m=1, b=5$. We denote different configurations with subscripts. The posterior draws are also obtained from 10\,000 MCMC iterations.

\vspace{-0.6cm}
\subsection{Scenario 1: A Single Document} 
\vspace{-0.3cm}
We first consider the case of a single document and the number of categories $K$  greater than the total counts $N$. We set $K=100$ and $N=50$. We examine four different settings here and the results are reported in  \Cref{tab:single_doc}. The first two settings are generated with fixed $\bpi$ and the latter two  are generated with fixed $\balpha$.

\begin{table}[!ht]
\centering
\resizebox{0.95\textwidth}{!}{
\begin{tabular}{l  *{11}{c}}
\toprule
MTD & DM            & Tuyl          & ZIDM          & PH$_1$-h        & PH$_2$-h    & PH$_3$-h        & PH$_4$-h       & PH$_1$-d       & PH$_2$-d     & PH$_3$-d       & PH$_4$-d       \\
\midrule
& \multicolumn{11}{c}{Setting 1: $\pi_{1}=\cdots=\pi_K=1/K$}\\
ABS$^{\times 100}$ &  0.697(0.02)  & 0.455(0.021) & 0.962(0.027) & 0.195(0.042) & 0.269(0.043) & \bf{0.163(0.034)} & 0.215(0.037) & 1.198(0.034) & 1.199(0.034) & 0.701(0.02)  & 0.741(0.02)  \\
COV & 0.999(0.004) & 1(0.002)     & 0.988(0.009) & 1(0)         & 1(0)         & 1(0)         & 1(0)         & 0.385(0.023) & 0.385(0.023) & 1(0.002)     & 1(0.002)     \\
    &              \multicolumn{11}{c}{Setting 2: $\pi_{k}=k/\sum_{i=1}^K i$}\\
ABS$^{\times 100}$ & 0.651(0.055) & 0.521(0.045) & 0.831(0.061) & 0.476(0.045) & 0.488(0.051) & \bf{0.469(0.036)} & 0.474(0.041) & 1.056(0.068) & 1.057(0.068) & 0.647(0.052) & 0.664(0.052) \\
COV & 0.993(0.01)  & 0.997(0.006) & 0.949(0.022) & 0.975(0.024) & 0.987(0.014) & 0.961(0.033) & 0.977(0.021) & 0.356(0.034) & 0.357(0.034) & 0.996(0.005) & 0.996(0.005) \\
    &      \multicolumn{11}{c}{Setting 3: $\alpha_{1}=\cdots=\alpha_K=1/K$}\\
ABS$^{\times 100}$ & 0.342(0.096) & 0.176(0.082) & 0.282(0.069) & 0.16(0.079)  & 0.16(0.079)  & 0.167(0.079) & 0.166(0.079) & 0.151(0.075) & \bf{0.15(0.075)}  & 0.274(0.095) & 0.406(0.112) \\
COV & 0.238(0.032) & 0.039(0.013) & 0.988(0.012) & 0.962(0.021) & 0.962(0.021) & 0.96(0.031)  & 0.958(0.033) & 0.137(0.035) & 0.134(0.018) & 0.424(0.043) & 0.561(0.054) \\
  &               \multicolumn{11}{c}{Setting 4: $\alpha_{k}=k/K $}\\
ABS$^{\times 100}$ & 0.668(0.05)  & \bf{0.642(0.049)} & 0.731(0.065) & 0.655(0.049) & 0.652(0.048) & 0.658(0.051) & 0.655(0.049) & 0.871(0.088) & 0.873(0.088) & 0.663(0.05)  & 0.661(0.053) \\
COV & 0.952(0.02)  & 0.97(0.019)  & 0.947(0.026) & 0.893(0.06)  & 0.903(0.043) & 0.873(0.066) & 0.881(0.057) & 0.314(0.038) & 0.311(0.032) & 0.969(0.017) & 0.975(0.017) \\
\bottomrule
\end{tabular}}
\caption{\small Performance comparison of different posteriors from a single document. The bold fonts mark the lowest ABV in each setting.}\label{tab:single_doc}
\vspace{-0.3cm}
\end{table}

From \Cref{tab:single_doc}, we observe that under fixed $\bpi$, the homogeneous PH prior with $m=1, b=3$ provides the best estimates. This is not surprising since the zero counts in these two settings are caused by insufficient sampling depth, i.e., small $N$, rather than structural zeros ($\pi_k=0$). The super-efficient shrinking parameter combination would not be helpful in this case. For Setting 3 when $\balpha$ is fixed and uniform, we observe that the heterogeneous PH priors with $m=0$ perform the best but at the cost of low coverage. The true $\alpha_k=1/K$ is close to zero and a strong shrinkage prior could better capture this sparsity pattern. For Setting 4, when $\balpha$ is heterogeneous, our heterogeneous prior did not outperform Tuyl's method. In addition, the sparsity-inducing configuration $(m=0)$ returns worse estimates than the configuration of $m=1$. Overall, we recommend using the homogeneous sparsity-inducing $(m=0, b=2)$ prior when there is only one document and the total number of count $N$ is small. While the heterogeneous prior can adapt to different count patterns more flexibly, it will fail to concentrate on the true values when there is not enough data to learn its many parameters.

\vspace{-0.6cm}
\subsection{Scenario 2: Multiple Documents}  
\vspace{-0.3cm}
We now consider  cases where there are $S$ documents $\bn_1, \ldots, \bn_S$ realized from different $\bpi_s$, where $\bpi_s \overset{\text{i.i.d.} }{\sim} \Dir(\balpha)$. The corresponding posterior distribution $p(\balpha\mid \bn_1, \bn_2, \ldots, \bn_S)$ can be written as
\begin{align}
p(\balpha\mid \bn_1, \ldots, \bn_S) &\propto  p(\balpha)\prod_{s=1}^S p(\bn_s\mid\balpha) \nonumber \\
& \propto \biggl( \prod_{k=1}^K \frac{1}{[c\alpha_k+a]^b}\biggr) \times \prod_{s=1}^S \biggl(  \frac{1}{[A]^{N_s}} \prod_{k=1}^K[\alpha_k]^{n_{sk}}  \biggr)\label{eq:multi_pi}
\end{align}
where $n_{sk}$ is the count of class $k$ in document $s$, $N_s=\sum_{k=1}^K n_{sk}$ is the total number of counts in document $s$ and  $A=\sum_{k=1}^K \alpha_k$. We choose $S=50$ and $K=100$. For each repetition, we draw $N_s$ uniformly from integers between $50$ and $150$ for every $s =1, \ldots, 50$. Similar to Scenario 1, we consider settings where true $\balpha$ is homogeneous and heterogeneous. The results are reported in \Cref{tab:multiple_docs}.

\begin{table}[!ht]
\centering
\resizebox{0.95\textwidth}{!}{
\begin{tabular}{l  *{11}{c}}
\toprule
MTD & DM            & Tuyl          & ZIDM          & PH$_1$-h        & PH$_2$-h    & PH$_3$-h        & PH$_4$-h       & PH$_1$-d       & PH$_2$-d     & PH$_3$-d       & PH$_4$-d       \\
\midrule
    &      \multicolumn{11}{c}{Setting 1: $\alpha_{1}=\cdots=\alpha_K=1/K$}\\
ABS$^{\times 100}$ & 0.152(0.009) & 0.141(0.009) & 0.138(0.009) &  \bf{0.134(0.008)} &  \bf{0.134(0.008)} &  \bf{0.134(0.008)} &  \bf{0.134(0.008)} &  \bf{0.134(0.008)} &  \bf{0.134(0.008)} & 0.141(0.009) & 0.141(0.009) \\
COV & 0.817(0.006) & 0.046(0.002) & 0.872(0.022) & 0.952(0.005) & 0.951(0.005) & 0.951(0.005) & 0.951(0.005) & 0.885(0.025) & 0.886(0.026) & 0.939(0.006) & 0.939(0.006) \\
  &               \multicolumn{11}{c}{Setting 2: $\alpha_{k}=k/K $}\\
ABS$^{\times 100}$  & 0.509(0.012) & 0.545(0.013) & 0.512(0.012) & 0.541(0.013) & 0.541(0.013) & 0.541(0.013) & 0.541(0.013) & 0.503(0.011) & \bf{0.502(0.011)} & 0.508(0.012) & 0.507(0.011) \\
COV & 0.94(0.004)  & 0.958(0.003) & 0.953(0.004) & 0.871(0.005) & 0.871(0.005) & 0.87(0.005)  & 0.87(0.005)  & 0.941(0.01)  & 0.943(0.01)  & 0.948(0.004) & 0.95(0.003) 
 \\
\bottomrule
\end{tabular}}
\caption{\small Performance comparison of different posteriors from multiple  documents. The bold fonts mark the lowest ABV in each setting.}\label{tab:multiple_docs}
\vspace{-0.3cm}
\end{table}

\begin{figure}[!ht]
\centering
\subfigure[$\PH(m=0, a=1, b=2, c=1)$]{
\includegraphics[width=0.95\textwidth]{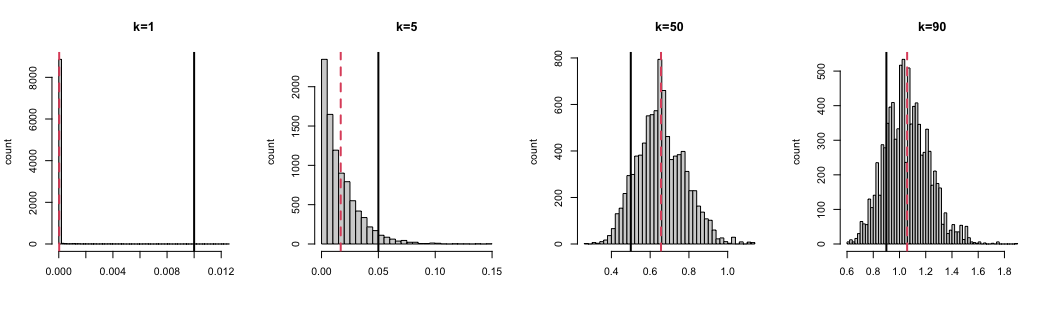}
}
\subfigure[$\PH(m=1, a=1, b=3, c=1)$]{
\includegraphics[width=0.95\textwidth]{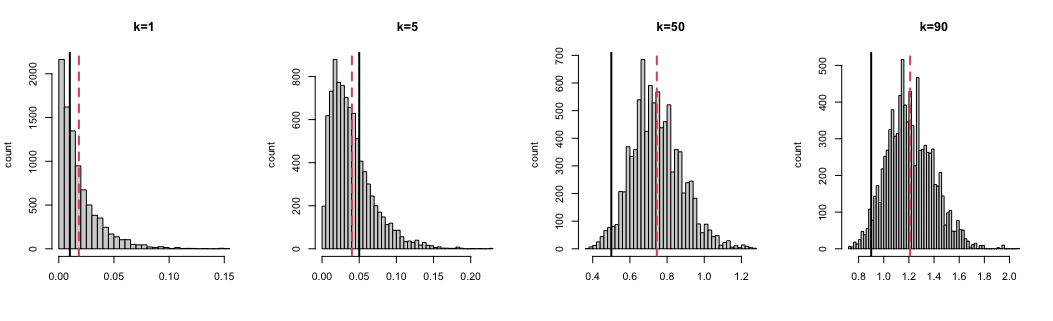}
}
\vspace{-5pt}
\caption{\small Posteriors of $\alpha_k$ from Setting 2. We include four $\alpha_k$ with $k=1, 5, 50, 90$. The black solid vertical line marks the location of true $\alpha_k=k/K$ and the red dashed vertical line marks the posterior mean. In this specific example, $N_1=0$.}\label{fig:posterior_multiple_docs}
\end{figure}

For \Cref{tab:multiple_docs}, we observe that with more information, our horseshoe prior  combined with heterogeneous modeling better captures the true values of $\balpha$, irrespective of whether the true $\balpha$ is homogeneous or heterogeneous. We include a plot of the posteriors under $\PH(m=0, a=1, b=2, c=1)$ prior and $\PH(m=1, a=1, b=3,c=1)$ from one repetition under Setting 2  in \Cref{fig:posterior_multiple_docs}, where the former prior induces bigger shrinkage than the latter. In this specific example, the first category $k=1$ has zero counts across all documents, i.e., $N_1=0$. Our prior effectively shrinks all posterior mass towards zero. For $k=5$, while the total number of counts is very small but non-zero, we observe a horseshoe-shaped posterior, indicating a strong shrinkage effect for small counts with $\PH(m=0, a=1, b=2, c=1)$, while the posterior mass shifts away from zero with the $\PH(m=1, a=1, b=3, c=1)$ prior. For $k=50$ and $k=90$, our  posterior distributions exhibit a bell shape and cover the true values. The plot suggests that the heterogeneous shrinkage prior $\PH(m=0, a=1, b=2, c=1)$ helps the posterior properly adapt to both sparse and quasi-sparse counts.

\vspace{-0.6cm}

\subsection{Scenario 3: Multiple Documents with Structural Zeros} 
\vspace{-0.3cm}
Here we consider cases where there are structural zeros in the dataset, i.e., a zero pattern that is shared by different documents and is different from events that have a positive probability but still observe a zero count. Note that our interpretation of structural zeros is different from the one in \citep{koslovsky2023bayesian}. Instead of a two-component representation, we translate the event $\pi_{sk}=0$ to be $\alpha_k=0$ which is consistent with the conventional DM representation. For the ``at-risk'' zeros, which corresponds to $n_{sk}=0$ but $\pi_{sk}>0$, we  believe that $\alpha_k > 0$. To enforce the zero patterns, we initialize $\alpha_k=k/K$ and then randomly select  a $q\%$ of the $\alpha_k$'s and let them be zero, so the category probabilities $\{\pi_{sk}\}$ sampled from these $\alpha_k$'s will always be zero. The results are reported in \Cref{tab:multiple_docs_zeros}.

\begin{table}[!ht]
\centering
\resizebox{0.95\textwidth}{!}{
\begin{tabular}{l  *{11}{c}}
\toprule
MTD & DM            & Tuyl          & ZIDM          & PH $_1$-h        & PH$_2$-h    & PH$_3$-h        & PH$_4$-h       & PH$_1$-d       & PH$_2$-d     & PH$_3$-d       & PH$_4$-d       \\
\midrule
    &      \multicolumn{11}{c}{Setting 1: $q=10\%$ }\\
ABS$^{\times 100}$ & 0.503(0.009) & 0.551(0.01)  & 16.203(0.974) & 0.547(0.01)  & 0.547(0.01)  & 0.547(0.01)  & 0.547(0.01)  & 0.494(0.009) & \bf{0.493(0.009)} & 0.5(0.009)   & 0.499(0.009) \\
COV & 0.846(0.004) & 0.959(0.003) & 0.959(0.003)  & 0.795(0.006) & 0.795(0.006) & 0.795(0.006) & 0.795(0.006) & 0.949(0.009) & 0.95(0.009)  & 0.953(0.003) & 0.954(0.003) \\
  &     \multicolumn{11}{c}{Setting 2:  $q=30\%$}\\
ABS$^{\times 100}$  & 0.457(0.014) & 0.513(0.016) & 17.514(1.084) & 0.514(0.016) & 0.514(0.016) & 0.514(0.016) & 0.514(0.016) & 0.444(0.013) & \bf{0.443(0.013)} & 0.451(0.014) & 0.45(0.014)  \\
COV & 0.656(0.004) & 0.962(0.003) & 0.966(0.003)  & 0.63(0.006)  & 0.63(0.006)  & 0.63(0.006)  & 0.63(0.006)  & 0.959(0.009) & 0.96(0.009)  & 0.964(0.003) & 0.965(0.003) \\
  &               \multicolumn{11}{c}{Setting 3:  $q=50\% $}\\
ABS$^{\times 100}$  & 0.41(0.013)  & 0.457(0.018) & 20.432(1.543) & 0.464(0.017) & 0.464(0.018) & 0.465(0.018) & 0.464(0.018) & 0.394(0.013) & \bf{0.393(0.013)} & 0.402(0.013) & 0.401(0.013) \\
COV & 0.468(0.003) & 0.964(0.003) & 0.974(0.003)  & 0.455(0.005) & 0.455(0.005) & 0.455(0.005) & 0.455(0.005) & 0.973(0.006) & 0.973(0.006) & 0.973(0.003) & 0.974(0.003) \\
\bottomrule
\end{tabular}}
\caption{\small Performance comparison of different posteriors from multiple  documents with structural zeros. The bold fonts mark the lowest ABV in each setting.}\label{tab:multiple_docs_zeros}
\end{table}

In \Cref{tab:multiple_docs_zeros}, we observe that as the percentage of structural zeros increases, the performance of ZIDM worsens, even though it provides valid coverage for the true probabilities, indicating wide credible intervals. Our choice of prior, $\PH(m=0, a=1, b=5, c=1)$, offers the best posterior mean estimates in all three settings, closely followed by the choice $\PH(m=0, a=1, b=2, c=1)$. Under the heterogeneous setting, all PH priors are able to provide valid coverage for true probabilities. 

Combing the results in \Cref{tab:multiple_docs} and \Cref{tab:multiple_docs_zeros}, we find that using the shrinkage  prior $\PH(m=0, a=1, b=5,c=1)$ with heterogeneous $\balpha$ yields optimal adaptability. While there is slight difference between priors $\PH(m=0, a=1, b=2, c=1)$ and $\PH(m=0, a=1, b=5, c=1)$, the performance overall is robust to the choice of $b$ and we recommend using the prior $\PH(m=0, a=1, b=2, c=1)$ for the strongest shrinkage effect. When there is insufficient information available, such as in the single document case in \Cref{tab:single_doc}, we recommend using the prior $\PH(m=0, a=1, b=2, c=1)$ with homogeneous $\alpha$.


\vspace{-0.7cm}
\section{Empirical Analysis}\label{sec:empirical}
\vspace{-0.3cm}
In this section, we provide an application of our method to  reduce sampling noise and understand the sparsity patterns in microbiome datasets.  Another example of how our method can be applied to analyzing Multiway Contingency Table is provided in \Cref{sec:contingency}.


The human microbiome comprises  microorganisms inhabiting both the surface and internal parts of  our bodies. Analyzing such data is challenging due to its compositional structure, over-dispersion, and zero-inflation. Consequently, the DM model and its variants have been widely utilized in this field, see for example \citet{holmes2012dirichlet,chen2013variable,wadsworth2017integrative,liu2020empirical,koslovsky2023bayesian}. Additionally, it is crucial to distinguish between structural zeros (indicating the absence of species) and sampling zeros (resulting from low sequencing depth or dropout). Our method provides a convenient approach to separate sampling noise from biological signals,  enhancing the accuracy and the interpretability of  downstream analyses such as microbial diversity studies, differential abundance testing \citep{liu2020empirical}, and the discovery of biomarkers. Reducing the sampling noise arising from low-sequencing depth is, in general, an important topic in omics data analysis. For example, \citet{wang2018gene}  use a  Poisson distribution to perform such deconvolution on single-cell RNA sequencing data. It helps infer properties of the gene expression distribution from raw counts, which improves many downstream analyses finding differentially expressed genes, identifying cell types, and selecting differentiation markers.  

Regarding the choice of DM or GDM, we choose to use heterogeneous DM with half-horseshoe PH prior in our analysis, while one can also use our half-horseshoe PH prior with GDM as we suggest in \Cref{sec:GDM}. Although the negative correlation structure imposed by DM models seems restrictive, DM models can be preferred over GDM models in certain cases thanks to its simplicity and interpretability, as it is able to capture overdispersion without introducing unnecessary complexity. Additionally, the extra parameters complicates the computation. \citet{tang2019zero} take an Expectation-Maximization approach for estimation which lacks intrinsic uncertainty quantification.

To showcase the usefulness of our method, we employ the human gut microbiome dataset studied in \citet{wu2011linking},  containing 28 genera-level operational taxonomic unit counts obtained from 16S rRNA sequencing on 98 objects. There are over 30\% zeros in the dataset.

The posterior mean estimates  for the category probabilities under the $\PH(m=0, a=1, b=2, c=1)$ prior for heterogeneous $\balpha$ is plotted in \Cref{fig:microbiome_PH}. Four taxonomic units (LIS, EIS, RIS, PIS) are shown in abbreviation. We compare our results with alternative methods used in the simulation studies in \Cref{sec:simulation} and find that the estimated probabilities are very similar. The average of absolute difference is below $1\times 10^{-4}$ between our method and ZIDM or Tuyl's. For this dataset, we observe that {\sl Bacteroides} predominates   in most individuals, and its lower concentration often coincides with a higher concentration of {\sl Prevotella}. 
The abundance levels of {\sl Lachnospiraceae Incertae Sedis (LIS), Subdoligranulum, Faecalibacterium, Alistipes, Parabacteroides, Peptostreptococcaceae Incertae Sedis (PIS)} show noticeable variations across individuals, which could be potentially related to variations in individual covariates. We hope to explore these relationships in future extensions of this work.
\begin{figure}[!ht]
\centering
\includegraphics[width=0.75\textwidth]{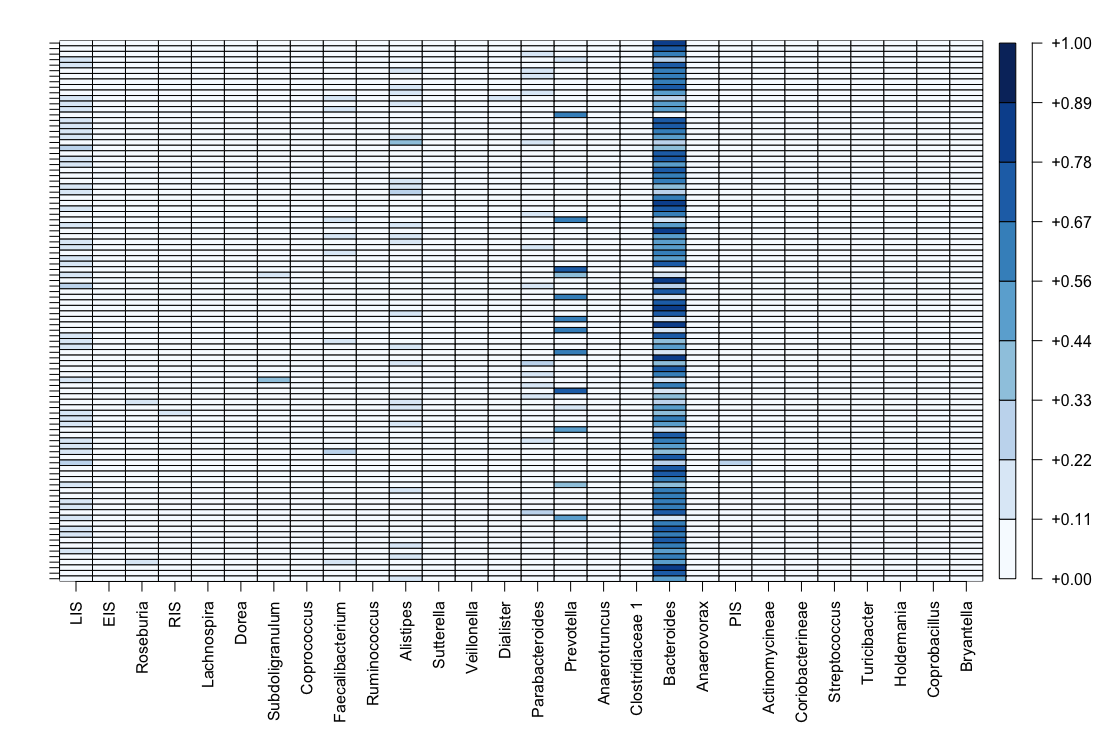}
\vspace{-12pt}
\caption{\small Posterior mean estimates for category probabilities from $\PH(m=0, a=1, b=2, c=1)$. Each row is a different individual and each column is a different taxonomic unit. }\label{fig:microbiome_PH}
\vspace{-0.5cm}
\end{figure}

\vspace{-0.7cm}
\section{Discussion}\label{sec:discussion}
\vspace{-0.3cm}
We introduce a novel class of prior applicable to a wide range of Bayesian count models. In particular, we demonstrate the motivation and the effectiveness of our method mainly using classical Dirichlet-Multinomial models. Our contributions are two-fold. First, rather than relying on empirical Bayes estimates or  maximum likelihood estimates \citep{minka2000estimating} for $\balpha$,  we propose a conjugate prior for $\balpha$ which enables full Bayesian inference for the DM models. The full posterior characterization of $\balpha$ provides better uncertainty quantification. In the case of homogeneous $\alpha$, we derive closed-form representations for the posterior density function and posterior moments can be obtained. The concentration parameter $\balpha$ plays an important role in subsequent analysis, including determining information sharing patterns in topic modeling such as LDA models.  Moreover, it was common practice to choose homogeneous DM models in LDA analysis due to the computational issues stemming from high-dimensional hyperparameters. Our PH priors provide an alternative to overcome these issues and thus potentially allow for a more flexible formulation.

Second, our priors provide a new approach for adapting to sparsity and quasi-sparsity in count datasets. When $m=0$, our priors exhibit a horseshoe-shaped behavior, which guarantees super-efficient sparsity recovery. Additionally, our representation obviates the complex two-groups mixture models and  retains the interpretable parameter $\balpha$. With the two features combined, our method is a superior alternative for reducing sampling noise in raw counts and producing better probability estimates.

There are several promising directions for future research. First, as discussed in \Cref{sec:other_models}, Pochhammer  distributions are conjugate to a number of  Bayesian count models, many of which face computational difficulties due to the Gamma ratio representation and the presence of excessive zeros and small nonzero counts. While we have outlined initial insights on how PH priors can be applied, further calibration is required for models beyond DM.  There are also many applications in nonparametric Bayes, such as the Pitman-Yor process.  Learning the full posterior for  the concentration parameter $\balpha$ enables us to capture complete uncertainty while maintaining analytical tractability. Second, our method has the potential to be extended to regression settings, allowing for sparsity patterns to be explained by covariates. This is especially relevant for microbiome studies, where researchers aim to uncover associations between varying abundances and many diseases. By allowing individual-specific concentration parameter $\alpha_{sk}$, we can relate $\alpha_{sk}$ to individual covariates such that it can learn information from both category level and individual level. 
Last, from a computational standpoint, it would be valuable to develop more scalable and efficient sampling strategies for our PH priors. 

\vspace{-0.5cm}
\spacingset{1.5}

\printbibliography[title={References}]
\end{refsection}

\clearpage

\appendix

\addcontentsline{toc}{section}{Appendix} 

\part{Appendix} 
\parttoc 

\begin{refsection}

\spacingset{1.9}

\section{Proofs}

\subsection{Proof of \Cref{thm:ph_formula}}\label{proof:ph_formula}

For brevity, let $\gamma_i$  represent $\gamma^{(m,a,b,c)}_i$ in the proof. The partial fraction expansion for the ratio of Pochhammer polynomials is given by
\[
 \frac{ [\alpha]^{m} }{ [  c\alpha + a ]^{b} } = \sum_{i=1}^{b} \frac{\gamma_i}{c\alpha + a + i -1 } \, \text{ where} \; \; \alpha \geq 0.
\]
Here $\gamma_i$ are residues and are determined by solving a set of linear equations in T\"oeplitz form using Levinson's algorithm. Specifically,  the residues $\gamma_i$ can
be calculated by carefully evaluating certain points of the identity (when $m\neq 0$)
\begin{equation}\label{eq:toeplitz_identity}
\prod_{s=1}^{m} ( \alpha + s -1 ) = \sum_{i=1}^{b} \gamma_i \prod_{k=1, \; k \neq i }^{b} (  c\alpha + a + k -1  ).
\end{equation}
If we evaluate at the points $\alpha= -(a+i-1)/c $, we get
\[
\gamma_j=\frac{\prod_{s=1}^{m} ( 1+(s-1)c-a-i ) }{c^m \prod_{k=1, \; k \neq i }^{b} (k-i)}.
\] 
This allows us to compute the normalizing constant as
\begin{align*}
C_{(m,a,b,c)}& =\sum_{i=1}^b \int_{0}^\infty \frac{\gamma_i}{c\alpha+a+i-1} \d \alpha\\
&= \sum_{i=1}^b \frac{\gamma_i}{c} \ln(c\alpha+a+i-1)\Big\lvert_0^\infty =- \sum_{i=1}^b \frac{\gamma_i}{c} \ln(a+i-1),
\end{align*}
which follows from term-by-term integration and the key property of the residuals that $\sum_{i=1}^b \gamma_i =0$, and
\[
\prod_{i=1}^{b} (  c\alpha + a + i -1 )^{\gamma_i} \rightarrow 1  \; {\rm and} \;
\sum_{i=1}^{b} (\gamma_i / c ) \ln (  c\alpha + a + i -1 ) \rightarrow 0 \; \; {\rm as} \; \; \alpha  \rightarrow \infty.
\]

When $m=0$, the RHS of \eqref{eq:toeplitz_identity} is just $1$, the rest of the calculations follows through.

\subsection{Proof of \Cref{thm:pph_formula}}\label{sec:proof_pph}

Similar to the proof of \Cref{thm:ph_formula}, we evaluate the identity when $m\neq 0$
\[
\alpha^d \prod_{s=1}^{m} ( \alpha + s -1 ) = \sum_{i=1}^{b} \gamma^{(m,a,b,c,d)}_i \prod_{k=1, \; k \neq i }^{b} (  c\alpha + a + k -1  )
\]
at  the set of points $\alpha=-(a+i-1)/c$, yielding
\[
\gamma^{(m,a,b,c,d)}_i= \frac{(1-a-i)^d\prod_{s=1}^{m} ( 1+(s-1)c-a-i ) }{c^{m+d} \prod_{k=1, \; k \neq i }^{b} (k-i)}.
\]

When $m=0$, similar we have the identity as
\[
\alpha^d = \sum_{i=1}^{b} \gamma^{(m,a,b,c,d)}_i \prod_{k=1, \; k \neq i }^{b} (  c\alpha + a + k -1  ),
\]
the rest of the calculation follows through.

The normalizing constant $C_{(m,a,b,c,d)}$ can be computed in a similar fashion to $C$ as
\begin{align*}
C_{(m,a,b,c,d)}= \sum_{i=1}^{b} \int_{0}^\infty \frac{\gamma^{(m,a,b,c,d)}_i}{c \alpha + a + i -1 } \d \alpha = \sum_{i=1}^b \frac{\Big(-\gamma^{(m,a,b,c,d)}_i\Big)}{c} \ln(a+i-1).
\end{align*}

\subsection{Double Roots}\label{sec:double_roots}

Consider a special case of the Pochhammer distribution with $a=0$
\[
p(\alpha\mid m, a=0,b,c)=\frac{1}{C_{m,0,b,c}}\frac{[\alpha]^m}{[c\alpha]^b}.
\]
From \Cref{thm:ph_formula}, we have 
\[
\gamma^{(m,0,b,c)}_i= \frac{\prod_{s=1}^{m} ( 1+(s-1)c-i ) }{c^m \prod_{k=1, \; k \neq i }^{b} (k-i)}, C_{(m,0,b,c)}=\sum_{i=2}^b \frac{\big(-\gamma^{(m,0,b,c)}_i\big)}{c} \ln(i-1) 
\]
From the first term,  we see $\gamma^{(m,o,b,c)}_1=0$, so the sum in normalizing constant can be reduce to $\sum_{i=2}^b$. The mean and variance of the prior can be calculated using \Cref{thm:pph_formula}.

\begin{theorem}[Posterior in Residues]\label{thm:ph_post1}
Under a Pochhammer prior $\alpha \sim \PH(m=n_0, a=0, b, c=K)$,  the posterior in \eqref{eq:alpha_cond_n} has closed-form as 
\begin{equation}\label{eq:alpha_post_double_roots}
p(\alpha \mid \bn)= C_N^{-1}\frac{ \prod_{k=0}^K[\alpha]^{n_k}}{[K\alpha]^{N}[K\alpha]^b} =  C_N^{-1} \left[\sum_{i=1}^{\max(N,b)}\frac{\gamma_i^*}{K\alpha+i-1}+\sum_{i=1}^{\min(N,b)} \frac{\beta_i^*}{(K\alpha+i-1)^2} \right]
\end{equation}
where 
{\small
\begin{align*}
\beta_i^*&=\frac{ \displaystyle \prod_{k=0}^K\prod_{s=1}^{n_k}\Big(1+ K(s-1)-i\Big)}{\displaystyle K^{N+n_0}\prod_{s=1,\; s\neq i}^{\min(N,b)}(s-i)^2\prod_{t=1}^{\abs{N-b}}(b+t-i)} \;\;\text{ for } 1\leq i \leq \min(N,b)\\
\gamma_i^* &= \left\{\begin{array}{l l}
{\displaystyle K\beta_i^*  \left(\sum_{k=0}^K\sum_{s=1}^{n_j=k}\frac{1}{1+K(s-1)-i } -  \sum_{s=1,\; s\neq i}^{\min(N,b)} \frac{2}{s-i}-\sum_{t=1}^{\abs{N-b}} \frac{1}{b+t-i}   \right)  }&  \text{ for } 1< i \leq \min(N,b)\\
\frac{ \displaystyle \prod_{k=0}^K\prod_{s=1}^{n_k}\Big(1+ K(s-1)-i\Big)}{ \displaystyle K^{N+n_0}\prod_{s=1,\; s\neq i}^{\min(N,b)}(s-i)^2\prod_{t=1}^{\abs{N-b}}(b+t-i)} &  \text{ for } \min(N,b)< i \leq \max(N,b)
\end{array}\right.\\
C_N&=\sum_{i=2}^{\max(N,b)}\frac{-\gamma_i^*}{K}\ln(i-1)+ \sum_{i=2}^{\min(N,b)} \frac{-\beta_i^*}{K(i-1)}.
\end{align*}}
\end{theorem}

\begin{proof}
Without loss of generality, we assume $b\leq N$. Using the same residual argument, we can write the identity as
\begin{equation*}
\begin{split}
\prod_{k=0}^K[\alpha]^{n_k}=& \sum_{i=1}^{N} \gamma_i^* (K\alpha+i-1)\prod_{s=1,\; s\neq i}^{b} (K\alpha+s-1)^2\prod_{t=1}^{N-b} (K\alpha+b+t-1) \\
&+ \sum_{i=1}^{b} \beta_i^* \prod_{s=1,\; s\neq i}^{b} (K\alpha+s-1)^2\prod_{t=1}^{N-b} (K\alpha+b+t-1).
\end{split}
\end{equation*}
Again, by evaluating the above identity at  $\alpha=-(i-1)/k$, we recover
\begin{align*}
\beta_i^*&=\frac{\prod_{k=0}^K\prod_{s=1}^{n_j}\Big(1+ K(s-1)-i\Big)}{K^{N+n_0}\prod_{s=1,\; s\neq i}^{b}(s-i)^2\prod_{t=1}^{N-b}(b+t-i)} \;\;\text{ for every } 1\leq i \leq b\\
\gamma_i^* &= \frac{\prod_{k=0}^K\prod_{s=1}^{n_k}\Big(1+ K(s-1)-i\Big)}{K^{N+n_0}\prod_{s=1}^{b}(s-i)^2\prod_{t=1, \; t\neq i-b}^{N-b}(b+t-i)} \; \; \text{ for every } b< i \leq N
\end{align*}
To calculate $\gamma_i^*$ for $1\leq i\leq b$, we use the method that
\begin{align*}
\gamma_i^*=\left.\left[\frac{\d}{\d \alpha} (K\alpha+i-1)^2 \frac{ \prod_{k=0}^K[\alpha]^{n_k}}{[K\alpha]^{N}[K\alpha]^b} \right] \right\vert_{\alpha=-(i-1)/K}
\end{align*}

We use $N(\alpha)$ and $D(\alpha)$ to denote the denominator and numerator functions, respectively, as
\begin{align*}
N(\alpha)&= \prod_{k=0}^K[\alpha]^{n_k} = \prod_{k=0}^K\prod_{s=1}^{n_k} (\alpha+s-1) \\
D(\alpha)& = \prod_{s=1,\; s\neq i}^{b} (K\alpha+s-1)^2\prod_{t=1}^{N-b} (K\alpha+b+t-1),
\end{align*}
then $\gamma_i^*$ can be written as
\begin{equation}\label{eq:gamma_ND}
\gamma_i^*= \left.\left[\frac{N'(\alpha)D(\alpha)-N(\alpha)D'(\alpha)}{ D(\alpha)^2} \right] \right\vert_{\alpha=-(i-1)/K}.
\end{equation}

Next, we break down the calculation using the product rule
\begin{align*}
N'(\alpha) &= N(\alpha) \sum_{k=0}^K\sum_{s=1}^{n_k}\frac{1}{\alpha+s-1 } \\
D'(\alpha)&= D(\alpha) \left(  \sum_{s=1,\; s\neq i}^{b} \frac{2k}{K\alpha+s-1}+\sum_{t=1}^{N-b} \frac{K}{K\alpha+b+t-1} \right).
\end{align*}

With the expression of $N'(\alpha)$ and $D'(\alpha)$ plugging into \eqref{eq:gamma_ND}, we have
\begin{align*}
\gamma_i^*&= \left.\left[\frac{N(\alpha)}{ D(\alpha)}\left(\sum_{k=0}^K\sum_{s=1}^{n_k}\frac{1}{\alpha+s-1 } -  \sum_{s=1,\; s\neq i}^{b} \frac{2K}{K\alpha+s-1}-\sum_{t=1}^{N-b} \frac{K}{K\alpha+b+t-1}   \right) \right] \right\vert_{\alpha=-(i-1)/K}.\\
&=\frac{\prod_{k=0}^K\prod_{s=1}^{n_k}\big(1+K(s-1)-i\big)}{K^{N+n_0}\prod_{s=1,\; s\neq i}^b (s-i)^2\prod_{t=1}^{N-b} (b+t-i)} 
\\
&\qquad \times \left(\sum_{k=0}^K\sum_{s=1}^{n_k}\frac{K}{1+K(s-1)-i } -  \sum_{s=1,\; s\neq i}^{b} \frac{2K}{s-i}-\sum_{t=1}^{N-b} \frac{K}{b+t-i}   \right) \\
&=\frac{\prod_{k=0}^K\prod_{s=1}^{n_j}\big(1+K(s-1)-i\big)}{K^{N+n_0-1}\prod_{s=1,\; s\neq i}^b (s-i)^2\prod_{t=1}^{N-b} (b+t-i)} 
\\
&\qquad \times \left(\sum_{k=0}^K\sum_{s=1}^{n_k}\frac{1}{1+K(s-1)-i } -  \sum_{s=1,\; s\neq i}^{b} \frac{2}{s-i}-\sum_{t=1}^{N-b} \frac{1}{b+t-i}   \right) \; \text{ for every } 1\leq i \leq b.
\end{align*}
Note that  we have $\gamma_1^*=\beta_1^*=0$, the calculation of the normalizing constant $C_N$ is similar to previous argument as
\begin{align*}
C_N&=\int_0^\infty  \sum_{i=1}^{N}\frac{\gamma_i^*}{K\alpha+i-1}+\sum_{i=1}^{b} \frac{\beta_i^*}{(K\alpha+i-1)^2} \; \d \alpha \\
&= \sum_{i=1}^N  \frac{(-\gamma_i^*)}{K}\ln(i-1) + \sum_{i=1}^b \frac{(-\beta_i^*)}{K(i-1)}.
\end{align*}

\end{proof}

\subsection{Closed-form Expressions for Conditional Density in \eqref{eq:ak_post}}\label{sec:heter_cond}

\begin{theorem}\label{thm:heter_cond}
Denote $A_{-k}=\sum_{j=1, j\neq k}^K \alpha_j=A-\alpha_k$, the Gibbs conditional  $\alpha_k$ given $A_{-k}$ can be written as
\begin{align*}
p(\alpha_k\mid n_k, N, A_{-k}) \propto \frac{1}{[\alpha_k + A_{-k}]^N}\frac{[\alpha_k]^{n_k}[\alpha_k]^{m}}{[c\alpha_k+a]^b}
\end{align*}
Using the same residual argument, the density function can be written explicitly as
\begin{equation}\label{eq:gibbs_heter}
p(\alpha_k\mid n_k, N, A_{-k}) = C_{n_k,N, A_{-k}}(m,a,b,c)^{-1} \left[ \sum_{j=1}^N \frac{ \gamma_j^*}{\alpha_k+A_{-k}+j-1} +\sum_{j=1}^b \frac{\beta_j^*}{c \alpha_k+a+j-1}\right]
\end{equation}
where 
\begin{align*}
\gamma_j^* &=  \frac{[-(A_{-k}+j-1)]^{n_k}[-(A_{-k}+j-1)]^{m}}{[-c(A_{-k}+j-1)+a]^{b}\prod_{t=s, s\neq j}^N (s-j)},\\
\beta_j^* &= \frac{c^{N-n_k-m}\prod_{s=1}^{n_k} (1+(s-1)c-a-j)  \prod_{s=1}^{m} (1+(s-1)c-a-j)}{\prod_{s=1, s\neq j}^b (t-j) \prod_{s=1}^N (1+(A_{-k}+s-1)c-a-j)},  \\
C_{n_k,N, A_{-k}} (m, a,b,c) &= \sum_{j=1}^N {(-\gamma_j^*)}\ln(A_{-k}+j-1)+ \sum_{j=1}^b \frac{(-\beta_j^*)}{c}\ln(\alpha_k+a+j-1).
\end{align*}
\end{theorem}
The proof still follows the same residual argument.

Here we also want to provide details on how we sample such conditional distribution using Metropolis-Hastings in \Cref{alg:post_heter}.

\begin{figure}[!ht]
\centering
\begin{minipage}{0.8\linewidth}
\begin{algorithm}[H]
	\caption{\em Posterior Sampling When $\alpha_k$ is heterogeneous.}\label{alg:post_heter}
\centering
	
\resizebox{0.8\linewidth}{!}{
		\begin{tabular}{l l}

\hline
			\multicolumn{2}{c}{\bf Input \cellcolor[gray]{0.6} }\\	
				\hline
			\multicolumn{2}{c}{Hyperparameters $m, a, b, c$, count data $\bn$, stepsize $\sigma$ }\\
				\hline
			\multicolumn{2}{c}{\bf  Sampling \cellcolor[gray]{0.6}}\\		
			\hline
			\multicolumn{2}{c}{Initialize  parameters  $\alpha_k^{(0)}\iid \text{logNormal} (0,1)$ for $k=1, \ldots, K$}\\
		
\multicolumn{2}{c}{  \cellcolor[gray]{0.9}{\bf Metroplis-Within-Gibbs} }\\
\multicolumn{2}{l}{For $t=1,\dots, T$:}\\
		
\quad For $k=1\ldots, K$: \\

\qquad Compute $A_{-k} =  A- \alpha_k^{(t-1)}$\\
	
		\qquad {Propose $\alpha_k^{(t+1)}= \exp(\log(\alpha_k^{(t-1)})+\epsilon)$ with $\epsilon\sim \mN(0,\sigma^2)$}\\
		
		\qquad {Compute log acceptance ratio $\text{LAR}=\log p(\alpha_k^{(t)}\mid \bn, A_{-k})-\log p(\alpha_k^{(t-1)}\mid \bn, A_{-k}) $}\\ 
			\qquad Generate   $u \sim \text{Unif}[0,1]$\\
			\qquad If $\log(u)<\text{LAR}$, set $\alpha_k^{(t)}=\alpha_k^{(t-1)}$, else $A=A+\alpha_k^{(t)}-\alpha_k^{(t-1)}$ \\			
{\bf Return} $\balpha^{(1)}, \ldots, \balpha^{(T)}$ & 
\end{tabular}}
\end{algorithm}
\end{minipage}
\vspace{-0.5cm}
\end{figure}

\subsection{Proof of \Cref{prop:heavy_tail}}\label{proof:heavy_tail}

We utilize Stirling's approximation of the Gamma function and the identity $\lim_{x\to \infty}(1+x/b)^x=e^b$. We have 
\[
\frac{1}{[cx+a]^b} =\frac{\Gamma(cx+a)}{\Gamma(cx+a+b)} \sim \frac{e^{-cx-a}(cx+a)^{cx+a-1/2}}{e^{-cx-a-b}(cx+a+b)^{cx+a+b-1/2}} \sim \frac{1}{(cx+a+b)^b}, \quad x\to \infty.
\]
Then we have
\[
\frac{x^m}{[cx+a]^b} \sim \frac{x^m}{(cx+a+b)^b} \sim \frac{1}{c^b x^{b-m}}, \quad x\to \infty,
\]
since $b\geq m+2$ by definition. Combing the equations above with L'H\^ opital's rule, we can write
\begin{align*}
\lim_{x\to\infty} e^{tx} P(\alpha>x)&= \lim_{x\to\infty} \frac{1}{C_{(m,a,b,c)}}\frac{\int_{x}^\infty \frac{\alpha^m}{[c\alpha+a]^b}\d\alpha}{e^{-tx}}\\
&= \lim_{x\to\infty} \frac{1}{C_{(m,a,b,c)}}\frac{ \frac{x^m}{[cx+a]^b}}{te^{-tx}} = \lim_{x\to\infty }\frac{1}{C_{(m,a,b,c)}} \frac{e^{tx}}{t c^bx^{b-m}} =\infty.
\end{align*}

\section{Other Count Models: Continued}\label{sec:other_models_p2}

\subsection{Ewens's Sampling Formula} \label{sec:ESF}
In population genetics, Ewens's Sampling Formula (ESF) \citep{ewens1972sampling} describes the probabilities associated with counts of how many different alleles are observed a given number of times in the sample. Suppose we have $n$ genes that can be summarized by its allelic partition $(m_1, m_2, \ldots, m_n)$, where $m_1$ is the number of alleles appearing exactly once, $m_2$ is the number of alleles appearing exactly twice, and so on. ESF assigns the probability of observing such partition as 
\begin{equation}\label{eq:ESF}
p(m_1, m_2, \ldots, m_n\mid \alpha) = \frac{n!}{[\alpha]^n}\prod_{j=1}^n \frac{\alpha^{m_j}}{j^{m_j} m_j!} 
\end{equation}
where $\alpha\geq 0$ represents the population mutation rate and $\sum_{j=1}^n j\cdot m_j=n$. When $\alpha=0$, it suggests that all $n$ genes are the same. When $\alpha=1$, the distribution corresponds to a uniform prior on all random permutations. As $\alpha\to \infty$, the probability that no two of the $n$ genes are the same going to $1$.

Inferring the mutation rate parameter $\alpha$ is not simple since the maximum likelihood estimation  (MLE) has no closed-form solution and MLE $\hat \alpha_n$ is asymptotically $\mN(\alpha, \alpha/\log n)$ \citep{ewens1972sampling}. The slow rate is a result of dependence among the genes and the effective sample size is roughly $\log n$. Other alternatives include Bayesian inference with a Gamma prior, method of moments, poisson process approximation \citep{arratia1992poisson}. Here we show that our Pochhammer prior enables fast full Bayesian inference.

By assigning a PH prior $\PH(m,a,b,c)$ to $\alpha$, we have
\begin{align*}
p(\alpha\mid m_1, m_2,\ldots, m_n) &\propto p(m_1, m_2, \ldots, m_n\mid \alpha) p(\alpha) \\
& \propto \frac{\alpha^{\sum {m_j}}[\alpha]^m}{[\alpha]^n[c\alpha+a]^b}.
\end{align*}
The posterior simplifies to $p(\alpha\mid m_1, m_2,\ldots, m_n) \propto \frac{\alpha^{\sum {m_j}}}{[\alpha+m]^{n-m}[c\alpha+a]^b}$ when $m<n$, which is typically the case. Using the same derivation as \Cref{thm:ph_post2}, we can provide closed-form expressions for the normalizing constant and posterior moments.

Discussion on the zero inflation in relation to ESF is scarce in the literature, compared to the other count models included in this paper. However, since ESF describes the distribution of allele frequencies in population, a common scenario is that many alleles are not observed in a sample due to their rarity. Certain alleles have zero counts even though they are present in the population, and the effect is more pronounced when the mutation rate $\alpha$ is low or when the sample size $n$ is small. Poisson process approximations \citep{arratia1992poisson} of ESF  have been explored for such situation.  Our Pochhammer prior might be helpful for modeling  zero-inflated allele counts in this context.

\subsection{Yule-Simon Distribution}

The Yule-Simon distribution \citep{yule1925ii,simon1955class} is commonly used to model phenomena with power-law behavior, such as word frequencies, or species abundances, where a small number of items are frequently observed, and most items occur sporadically.  Its probability mass function is given by
\[
p(n\mid \alpha) = \alpha B(n, \alpha+1) =\alpha \frac{\Gamma(n)\Gamma(\alpha+1)}{\Gamma(\alpha+n+1)}= \frac{\alpha (n-1)!} {[\alpha+1]^n}
\]
where $\alpha>0$ is the shape parameter controling tail heaviness. Due to its heavy tail, the probability of observing large $n$ decrease slowing, allowing for occasional large values to occur.

If we have multiple observation $n_1, n_2, \ldots, n_K$, the joint likelihood is 
\begin{equation}\label{eq:yule-simon}
p(n_1, n_2, \ldots, n_K\mid \alpha) = \prod_{k=1}^K \alpha B(n_k, \alpha+1) \propto  \frac{\alpha^K}{\prod_{k=1}^K[\alpha+1]^{n_k}}.
\end{equation}
By assigning a Pochammer prior $\PH(m,a,b,c)$ to $\alpha$,  and combining it with \eqref{eq:yule-simon}, the posterior becomes
\[
p(\alpha\mid n_1, n_2, \ldots, n_K) \propto \frac{\alpha^K[\alpha]^m}{[c\alpha+a]^b\prod_{k=1}^K[\alpha+1]^{n_k}}.
\]
We can employ the residual approach to derive closed-form expressions for the normalizing constant and posterior moments. Moreover, since the PH prior is also heavy tailed, we can preserve the power-law nature of the Yule-Simon distribution.

Although the Yule-Simon distribution does not have the zero-inflation phenomena since it starts counting at $n=1$, its generalizations, such as the two-parameter Waring distribution \citet{irwin1968generalized} and the three-parameter Beta Negative Binomial distribution, start counting at $n=0$ and zero inflation is common\citet{rivas2023zero}. Our prior can provide a conjugate updating scheme and inject more flexibility to accommodate zero counts in  these  models.  For the Waring distribution, which can be viewed as mixture of geometric distribution with the success rate parameter following from a Beta distribution, our PH priors can function on the Beta-distributed mixing. Similarly for the Beta Negative Binomial distribution, our prior can be applied to either the beta mixing part, the negative binomial mixture part, or even both. We leave the detailed investigation of these possibilities for future research.

\section{Application: Multiway Contingency Table} \label{sec:contingency}
\vspace{-0.3cm}
The Dirichlet distribution  also serves as a classical conjugate prior choice for analyzing multiway contingency tables, as demonstrated in \citet{good1976application}. However, similar to other applications using Dirichlet priors, the performance of the method is highly dependent on the choice of the concentration parameter, and a clear principle way of selecting $\alpha$ is lacking.

Moreover, when the number of categorical variables increases, the number of possible cells increases exponentially, leading to an enormous number of empty cells. For instance, in gene sequencing datasets, each position can be one of the \{A, C, G, T\} nucleotides. If data is collected from $p$ positions, then the possible number of cells in the multiway contingency table is $K=4^p$, which often far exceeds the number of available sequences. 

Instead of resorting to techniques like low-rank tensor decomposition \citep{zhou2015bayesian} or latent mixture models such as  \citet{dunson2009nonparametric}, we directly model the cell probabilities with our sparsity-inducing priors. 

For illustrative purposes, we use the {\sl E. coli} promoter gene sequence data  that is publicly available in R package \texttt{DMRnet}. We focus on the 53 promoter sequences and we examine the association among the first 7 positions. Thus in this case, we have $N=53$ and $K=4^7=16\,384$, with more than 99\% of the cells  empty. We then employ the posteriors of cell probabilities to explore the associations between different positions using Cramer's V statistics, ranging from $0$ (no association) to $1$ (perfect association). Let $\rho_{jj'}$ represent Cramer's V statistics between positions $j$ and $j'$, defined as follows
\[
\rho^2_{j j'}= \frac{1}{\min\{d_j ,d_{j'}\}-1} \sum_{c_j=1}^{d_j}\sum_{c_j'=1}^{d_{j'}} \frac{\pi_{c_j c_j'}^{(j, j')} - \pi_{c_j}^{(j)}\pi_{c_{j'}}^{(j')}}{\pi_{c_j}^{(j)}\pi_{c_{j'}}^{(j')}},
\]
where $\pi_{c_j}^{(j)}$ is the marginal probability of position $j$'s being $c_j$ and $\pi_{c_j c_j'}^{(j, j')}$ is the joint probability of that position $j$ being $c_j$ and position $j'$ being $c_{j'}$. In our analysis, all $d_j=4$. 

\Cref{fig:multiway_table} reports the posterior means and the 95\% credible intervals (CIs) of Cramer's V statistics. We observe that the lower 2.5\% CI of pairs $\{1,6\},\{3,6\}, \{3,7\}, \{4,6\}$  are above zero, which suggests there is a non-negligible association in those pairs. Among these four pairs, pair $\{4,6\}$ shows the strongest association, followed by pairs $\{1,6\}, \{3,6\}$, as reflected  in their posterior means and upper 97.5\% CIs.
\begin{figure}[!ht]
\subfigure[Posterior Mean]{
\includegraphics[width=0.31\textwidth]{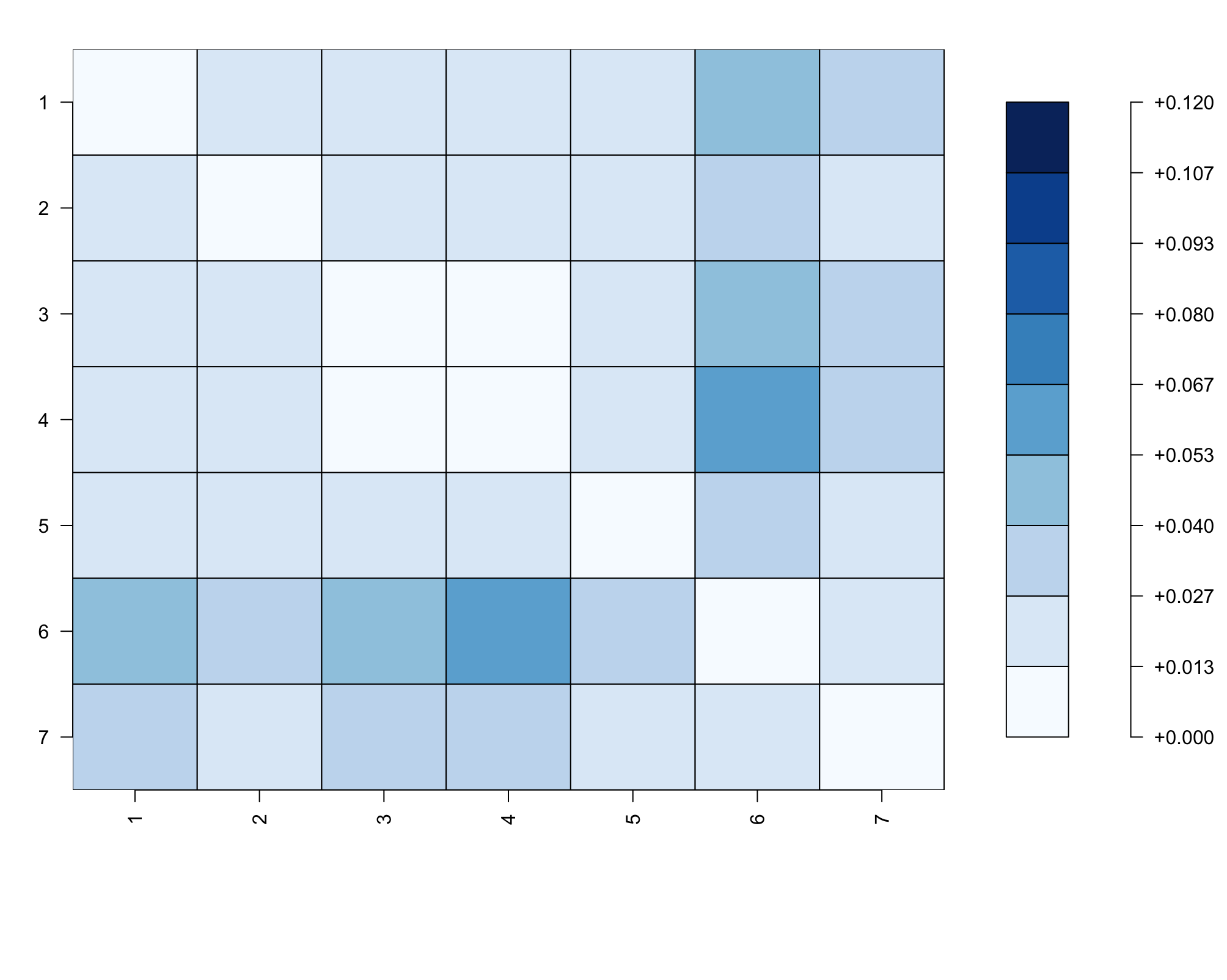}
}
\subfigure[Lower 2.5\% CI]{
\includegraphics[width=0.31\textwidth]{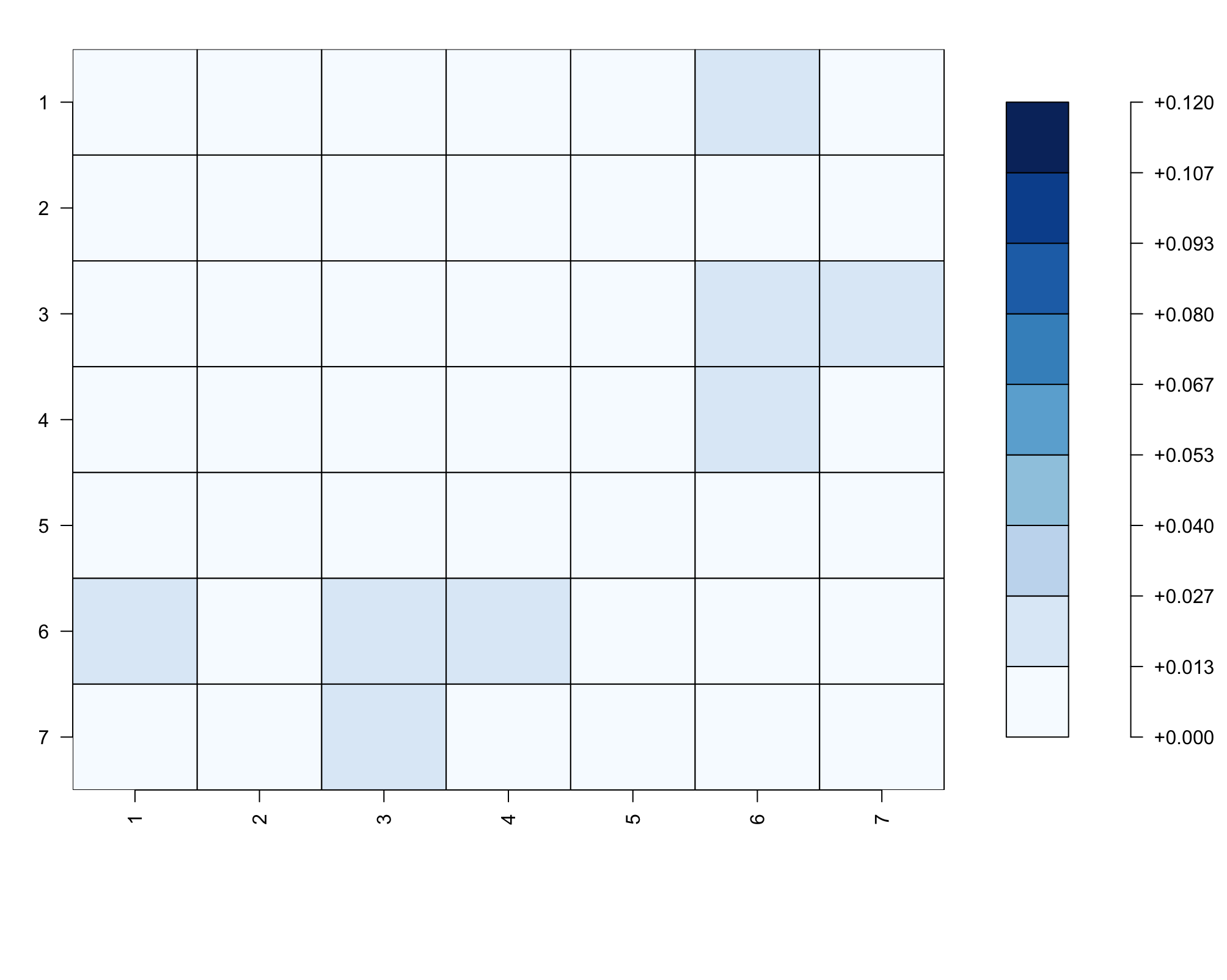}
}
\subfigure[Upper 97.5\% CI]{
\includegraphics[width=0.31\textwidth]{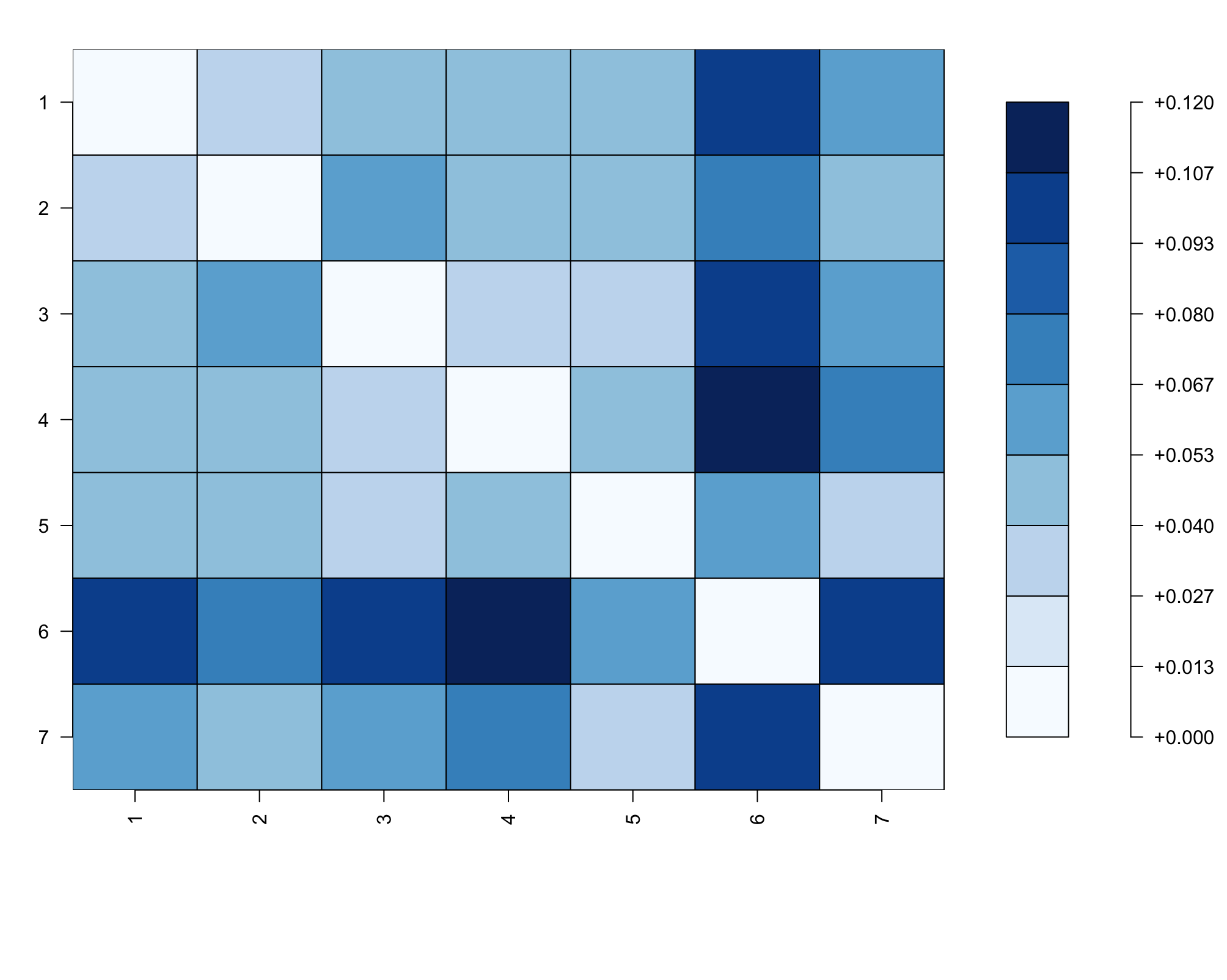}
}
\caption{Cramer's V Statistics for the first 7 positions} \label{fig:multiway_table}
\end{figure}

While our method may not yet match the scalability of previous approaches that impose low-dimensional structures, our method provides a more straightforward way to model the cell probabilities with results that are easier to interpret.

\printbibliography[heading=bibintoc,title={Appendix Refereces}]
\end{refsection}

\end{document}